\documentclass[11pt, a4paper]{article} 
\usepackage[english]{babel}
\usepackage[utf8]{inputenc}
\usepackage[a4paper]{geometry}

\usepackage{color}
\usepackage{graphicx}

\usepackage{mathrsfs}

\usepackage{float}

\usepackage{bbm}
\usepackage{enumerate}
\usepackage{tikz}
\usetikzlibrary{matrix}
\usepackage{pgfplots}
\usetikzlibrary{arrows.meta}
\usepackage{subcaption}

\usepackage[colorlinks=false,linkcolor=black,citecolor=black]{hyperref}

\usepackage[symbol]{footmisc}

\usepackage{amsmath}
\usepackage{amssymb}
\usepackage{amsthm}
\usepackage{stmaryrd}

\newcommand{\F}{\mathcal{F}}
\newcommand{\PM}{\mathbb{P}}

\newcommand{\E}{\mathbb{E}}

\newsavebox\dotbox
\sbox{\dotbox}{\(\displaystyle\bigodot\)}

\newcommand{\beao}{\begin{eqnarray*}}
\newcommand{\eeao}{\end{eqnarray*}\noindent}

\newcommand{\beam}{\begin{eqnarray}}
\newcommand{\eeam}{\end{eqnarray}\noindent}

\def\bbr{{\Bbb R}}   
\def\bbn{{\Bbb N}}

\newcommand{\eps}{{\varepsilon}}

\newcommand{\ov}{\overline}
\newcommand{\un}{\underline}
\newcommand{\wh}{\widehat}
\newcommand{\wt}{\widetilde}

\theoremstyle{theorem}
\newtheorem{theorem}{Theorem}[section]
\newtheorem{lemma}[theorem]{Lemma}
\newtheorem{corollary}[theorem]{Corollary}
\newtheorem{proposition}[theorem]{Proposition}

\newtheorem{note}[theorem]{Note}
\newtheorem{remark}[theorem]{Remark}
\newtheorem{example}[theorem]{Example}
\newtheorem{assumption}[theorem]{Assumption}
\theoremstyle{definition}
\newtheorem{definition}[theorem]{Definition}

\numberwithin{equation}{section}

\title{Optimal investment under capital gains taxes}

\author{Alexander Dimitrov\thanks{Institute of Mathematics, Goethe University Frankfurt, D-60054 Frankfurt a.M., Germany, e-mail: \{dimitrov, ckuehn\}@math.uni-frankfurt.de
\newline 
{\em Acknowledgments.} We would like to thank an anonymous referee for valuable comments and suggestions that led to a substantial improvement of the paper. During the revision, we used ChatGPT for an additional check of our mathematical arguments for errors.} 
\and Christoph K\"uhn\footnotemark[1]}

\date{}

\begin{document}
	\maketitle
	\begin{abstract}
We generalize classical existence results for expected utility maximization in discrete time frictionless market models to models with capital gains taxes.
We consider the realistic but mathematically challenging rule that losses do not trigger negative taxes but can only be offset against potential gains in the future. 
Central to the analysis is a well-known phenomenon from arbitrage-free markets with proportional transaction costs that does not exist in arbitrage-free frictionless markets: an investment in specific quantities of stocks that is completely riskless but may provide an advantage over holding money in the bank account. 	
As a result of this phenomenon, on an infinite probability space, no-arbitrage does not imply that the set of attainable terminal
wealth is closed in probability. 
We provide simple sufficient conditions for closedness. Then, we characterize the closure of the set of attainable terminal wealth, thereby identifying precisely the source of non-closedness.

As a by-product, we obtain a new construction for an integrable majorant that dominates the utilities of all nonnegative terminal wealth attainable from a given initial capital in a frictionless market and that works directly in multiperiod models.
\end{abstract}

\begin{tabbing}

\end{tabbing}	

\section{Introduction}

Capital gains taxes have a major impact on investors' returns, and their influence on optimal investment decisions has been well studied
in the financial economics literature. 
In an influential early contribution, Dybvig and Koo~\cite{dybvig.koo.1996} analyze a multiperiod utility maximization problem under the so-called {\em exact tax basis} or {\em specific share identification method}. This corresponds to the tax legislation in many countries, especially in the U.S., and seems economically to be the most reasonable tax basis: when an investor wants to reduce a position, she 
explicitly identifies which shares---purchased at which time and at which price---are being sold and thus relevant for taxation.
The model was further developed by Dammon, Spatt, and Zhang~\cite{dammon.spatt.zhang.2001, dammon.spatt.zhang.2004}, among many others, using the average tax basis and including a tax-deferred account. The average tax basis simplifies tax payments by considering an average purchasing price of identical shares in the portfolio. 

In the earlier financial economics literature, tax systems with the {\em full use of losses~(FUL)} have been considered, that is, overall losses in a year lead to negative tax payments, so-called tax credits. More realistic is the {\em limited use of losses~(LUL)}, that is, losses do not trigger negative taxes but can only be offset against potential gains in the future. Optimal portfolios with LUL compared to FUL are analyzed by Ehling, Gallmeyer, Srivastava, Tompaidis, and Yang~\cite{ehling.et.al.2018} and Fischer and Gallmeyer~\cite{fischer.gallmeyer.2017}.
Haugh, Iyengar, and Wang~\cite{haugh.iyengar.wang.2016} apply the duality method based on relaxing the adaptedness constraint of stochastic controls
to obtain dual upper bounds for the value of a utility maximization problem. The method allows the evaluation of heuristic strategies in problems with a 
rather large number of assets and time periods.   

Moreover, several advanced continuous time portfolio optimization problems with taxes have been analyzed:
In their pioneering work, Jouini, Koehl, and Touzi~\cite{jouini.koehl.touzi.1999, jouini.koehl.touzi.2000} derive first-order conditions for  
optimal consumption plans with the {\em first in first out~(FIFO)} tax basis and a quite general tax code, but only with a deterministic stock price process.
In Ben Tahar, Soner, and Touzi~\cite{tahar.soner.touzi.2007, tahar.soner.touzi.2010} the continuous time Merton problem for the average tax basis, FUL, and non-zero transaction costs is analyzed. It is shown that the value function is the unique viscosity  solution of the corresponding
dynamic programming equation.
Cai, Chen, and Dai~\cite{cai.chen.dai.2018} and
Bian, Chen, Dai, and Qian~\cite{bian.chen.dai.qian.2021} consider the same problem without transaction costs. They show that the value function is the minimal viscosity solution of the corresponding dynamic programming equation and approximate it by value functions of the problem with constrained trading rates. For the original problem, the analysis is based on $\eps$-optimal strategies.
In Seifried~\cite{seifried.2010} the continuous time optimization problem is solved for the LUL, but for a tax-deferred account, that is, taxes are only paid at maturity.

Nevertheless, a comprehensive mathematical theory for market models with taxes, comparable to those for frictionless markets and for markets with proportional transaction costs, is still lacking. This paper aims to contribute to closing this gap and to relate tax models to transaction costs models. Of course, tax legislation varies from country to country, but there are two core principles: taxes are realization-based, that is, they are paid
when gains are realized by selling securities and not when they arise on paper; and there is a limited use of losses~(LUL) as described above. 
We consider a model that captures this, but for simplicity, we abstract from many other detailed tax regulations that would distract from the mathematical core of the problem. We work with the exact tax basis, which is a natural choice for a theoretical analysis.

In a model with the full use of losses~(FUL), the after-tax wealth is affine linear in the trading strategy, specifying the number of stocks in the portfolio, and the model can be formally expressed as a transaction costs model by issuing at each point in time new artificial securities (see K\"uhn~\cite[Section~5]{kuehn.2019}). But LUL, which we consider in the current paper, makes the after-tax wealth non-linear in the trading strategy, which is why many standard arguments from the theory of proportional transaction costs are not applicable. On the other hand, a useful property of LUL is that the after-tax wealth cannot exceed the wealth of the same trading strategy but without taxes. This estimate plays a key role in the paper. In addition, under LUL, the no-arbitrage property is equivalent to that in the corresponding frictionless market. By contrast, with FUL, the after-tax wealth may exceed that of a tax-exempt investor, and the no-arbitrage property is even stronger than that under no taxes (cf. \cite[Remark~2.2]{kuehn.2019}).

For portfolio optimization, which is the main subject of the present paper, a key requirement is that the set of attainable terminal wealth is closed in probability. In discrete time frictionless models, this closedness already follows from no-arbitrage 
(Dalang, Morton, and Willinger~\cite{dmw.1990}). By contrast, in models with taxes, no-arbitrage alone does not imply that the set of attainable terminal wealth is closed in probability, and consequently, a finite supremum of a utility maximization problem need not be attained (Example~\ref{example:Nonclosedness}). 
Here, one runs into a phenomenon well known from models with proportional transaction costs. Let us describe this phenomenon in detail. 
First note that both transaction costs and taxes tend to make portfolio rebalancing disadvantageous. In the first case, an investment needs time to amortize trading costs, and in the second case, realizing book profits can trigger (earlier) tax payments. Now, there can exist an investment in specific quantities of stocks
that, in the first period, leads to the same pre-tax return as the bank account for sure.
On the one hand, a debt-financed purchase of these stocks is riskless in the first period (both for a tax-exempt and a taxable investor).
On the other hand, this purchase can lead to unrealized book profits and realized losses of the same amount at time~$1$. In the future,
this could provide an advantage for the taxable investor over doing nothing in the first period, since it could help her to defer or even avoid paying taxes. Consequently, with limited initial information, the investor can buy arbitrarily many of these stocks financed with debt at time~$0$ and then, based on additional information, sell the shares that are not needed at time~$1$---without any risk. On an infinite probability space, this can destroy the closedness of the set of attainable terminal wealth. The reason is that the amount of purchases at time~$0$ that could provide a tax advantage need not be bounded ex ante. Exactly the same phenomenon was observed by Schachermayer~\cite[Example~3.1]{schachermayer2004fundamental} in the context of arbitrage-free transaction costs models. By contrast, in an arbitrage-free frictionless market, the phenomenon cannot occur since there, an ``advantage'' boils down to a higher wealth expressed in terms of monetary units.
Consequently, in discrete time frictionless models, no-arbitrage already implies that the set of attainable terminal wealth is closed in probability (Dalang, Morton, and Willinger~\cite{dmw.1990}).

In the present paper, we provide sufficient conditions for closedness in the LUL tax model: On a finite probability space, the set of attainable terminal wealth is always closed---with respect to the supremum norm and even if the model would allow for an arbitrage~(Proposition~\ref{28.11.2025.1.ck}).
Another sufficient condition that can easily be verified is that no-arbitrage holds {\em and} there are no redundant stocks (Proposition~\ref{31.8.2026.1}). The no-redundancy assumption can be weakened to the requirement 
that any terminal wealth a taxable investor can achieve she can 
achieve without debt-financed redundant stocks (Proposition~\ref{31.8.2026.2}).

In addition, we characterize the closure of the set of attainable terminal wealth. We slightly extend the class of strategies by certain artificial ones:
we allow for an anticipative selection from any sequence of dominating strategies in the original model.
Then, we show that any random variable in the closure of the original set is the terminal wealth of an extended strategy~(Theorem~\ref{main:closedness}).
This means that closedness with the original strategies fails if and only
if there exists an anticipative selection from a sequence of dominating strategies in the original class that cannot be
dominated by another strategy from the original class.
The result allows us to identify precisely the source of the lack of closedness.

{\em Assuming} this closedness property, we analyze the utility maximization problem of terminal wealth. We work with a utility function that takes the value~$-\infty$ for a negative terminal wealth, which seems to be quite natural. The utility function has to satisfy only minimal assumptions, that is, it has to be nondecreasing and concave, but it need not even be differentiable. First, we show that the value function is finite if and only if the value function of the corresponding problem without taxes is finite~(Proposition~\ref{lemma:IntegrabilityAllx}). 
Then, under no-arbitrage and the additional assumptions that the set of attainable terminal wealth is closed in probability and the value function is finite, we show that an optimal strategy always exists~(Theorem~\ref{Theo:MainUtility}). This generalizes R\'asonyi and Stettner~\cite[Theorem~1.1]{rasonyi.stettner.2006} for frictionless models. Example~\ref{30.1.2026.1} shows that the optimal strategy need not be unique, even if the utility
function is strictly concave and the corresponding frictionless market has no redundant stocks.
We note that in continuous time, Kramkov and Schachermayer~\cite[Example~5.2]{kramkov.schachermayer.1999} provide a counterexample showing that the closedness of the set of attainable terminal wealth---there, it is guaranteed by NFLVR---does not imply that the finite supremum over expected utilities is attained. The counterexample justifies the stronger assumptions on the utility function made for the continuous time maximization problem in \cite{kramkov.schachermayer.1999}.

As a by-product, we generalize R\'asonyi and Stettner~\cite[Theorem~1.1]{rasonyi.stettner.2006} to frictionless markets with short-selling constraints. 
We obtain a different integrable majorant for the utilities of attainable terminal wealth, that may be of independent interest:
Let $x>0$ and let $d\in\bbn$ denote the number of risky assets. 
Then, in an arbitrage-free one-period model, there exists a trading strategy with initial capital~$(2^{d+1}-1)x$ whose terminal wealth dominates $P$-a.s. the terminal wealth of {\em any} strategy with initial capital~$x$ and $P$-a.s. nonnegative wealth. The result directly extends to a $T$-period model, where the required initial capital increases to $(2^{d+1}-1)^T x$.
Since the value function is finite for all positive initial capitals if it is finite for initial capital~$x$, the utility of the dominating strategy 
is an integrable majorant of all utilities that can be achieved with initial capital~$x$.
The proof in \cite[Theorem~1.1]{rasonyi.stettner.2006} relies on a similar one-period estimate, but with an integrable majorant which need not be the utility of an attainable terminal wealth for some initial capital. This is the reason why their estimate cannot be directly extended to a multiperiod model. To make their result applicable to the multiperiod case, \cite{rasonyi.stettner.2006} establish a dynamic programming principle and consider one-period majorants with regard to the value function instead of the utility function.\\

The rest of the paper is organized as follows. In Section~\ref{23.12.2025.1.ck}, we introduce the model and establish basic properties.
Section~\ref{29.12.2025.3} deals with utility maximization and contains the first main result (Theorem~\ref{Theo:MainUtility}). Section~\ref{23.8.2026.1.ck} provides sufficient conditions
for closedness and presents the second main result (Theorem~\ref{main:closedness}).

\section{The model}\label{23.12.2025.1.ck}

Throughout the paper, we fix a terminal time~$T\in\bbn$ and a filtered probability \linebreak space~$(\Omega,\mathcal{F},(\mathcal{F}_t)_{t=0,1,\ldots,T},P)$. (In)equalities between random variables are understood almost surely unless stated otherwise. In contexts involving random sets, however, ``a.s.'' is written explicitly. 
There are $d\in\bbn$ non-shortable risky assets with adapted price processes~$S^j=(S^j_t)_{t=0,1,\ldots,T}$, $j=1,\ldots,d$. In addition, the investor can both lend to and borrow from a bank account at the same interest rate modeled by the {\em nonnegative} adapted process~$r=(r_t)_{t=1,\ldots,T}$. That is, $r_t$ is the interest rate between $t-1$ and $t$. It seems to be reasonable to assume that $r_t$ is even $\mathcal{F}_{t-1}$-measurable, but mathematically this is not used anywhere in the paper.
Following the notation in \cite{dybvig.koo.1996}, $N_{i,t,j}\in L^0(\mathcal{F}_t; \bbr_+)$ denotes the number of assets of type~$j\in\{1,\ldots,d\}$ that are bought at time $i\in\{0,\ldots,T-1\}$ and kept in the portfolio at least after trading 
at time $t\in\{i,\ldots,T\}$. Here, $L^0(\mathcal{F}_t; \bbr_+)$ denotes the space of equivalence classes of $\F_t$-measurable $\bbr_+$-valued random 
variables, equipped with the topology of convergence in probability. 
By the implicit assumption that an asset cannot be bought and resold at the same time, $N_{i,i,j}$ is the number of shares of type~$j$ purchased at time~$i$. Short-selling is excluded, and one has the constraint
\beao
 N_{t,t,j}\geq N_{t,t+1,j}\geq\ldots\geq N_{t,T,j}=0\quad\mbox{for\ }t=0,\ldots,T-1,
\eeao
which forces liquidation at $T$. The set of all such strategies~$N$ is denoted by $\mathcal{N}$.
For brevity we sometimes write $S$ for the vector~$(S^1,\ldots,S^d)$, $N_{i,t}$ for the vector~$(N_{i,t,1},\ldots,N_{i,t,d})$, and denote by $\langle\cdot,\cdot\rangle$ and $|\cdot|$ the scalar product and the Euclidean norm in $\bbr^d$, respectively. By $\eta_t\in L^0(\mathcal{F}_t; \bbr)$ we denote the number of monetary units {\em after} paying taxes and trading at time~$t$. First, we define the accumulated realized profits and losses by
\beam\label{13.11.2025.1.ck}
G_0:=0,\quad G_t:= \sum_{u=1}^t\left(\eta_{u-1} r_u + \sum_{i=0}^{u-1}\langle N_{i,u-1}-N_{i,u},S_u-S_i\rangle\right),\quad t=1,\ldots,T.
\eeam
This means that taxes on the interest payments of the bank account cannot be deferred.
Then, the tax payment stream with limited use of losses~(LUL) is defined by
\beam\label{13.11.2025.2.ck}
\Pi_t := \alpha\max_{u=0,1,\ldots,t} G_u,\quad t=0,\ldots,T,
\eeam
where $\alpha\in[0,1)$ is the constant tax rate. A strategy~$(\eta,N)$ is called self-financing for the initial capital~$x\in\bbr$ iff $\eta_{-1}=x$ and 
\beam\label{13.11.2025.3.ck}
\eta_t - \eta_{t-1} =  r_t \eta_{t-1}1_{(t\ge 1)}  
+ \left\langle\sum_{i=0}^{t-1}(N_{i,t-1}-N_{i,t}) - N_{t,t}, S_t\right\rangle - (\Pi_t-\Pi_{t-1})1_{(t\ge 1)}
\eeam
for $t=0,\ldots,T$. We note that given $x\in\bbr$ and $N\in\mathcal{N}$, definition~(\ref{13.11.2025.1.ck})-(\ref{13.11.2025.3.ck}) yields an explicit construction of a unique adapted process~$\eta=:\eta(x,N)$ such that $(\eta,N)$ is self-financing with initial capital~$x$. This is because $G_t$ and thus $\Pi_t$ do not depend on $\eta_t$. 
For every $\alpha\in[0,1)$, the terminal wealth is denoted by
\beao
V^\alpha(x,N):=\eta_T.
\eeao
For the special case of a frictionless market, i.e., $\alpha=0$, the wealth {\em process} can be defined by
\beao
V^0_t(x,N):=\eta_t + \left\langle\sum_{i=0}^t N_{i,t}, S_t\right\rangle,\quad t=0,1,\ldots,T.
\eeao
It satisfies the recursion $V_0^0(x,N)=x$ and 
\begin{equation}\label{eq:frictionlesswealth}
    V_t^0(x,N)=(1+r_t)V_{t-1}^0(x,N)+\left\langle\sum_{i=0}^{t-1}N_{i,t-1},S_t-(1+r_t)S_{t-1}\right\rangle,\quad t=1,\ldots,T.
\end{equation}
By contrast, for $\alpha>0$, there is no canonical one-dimensional wealth for $t<T$, as there is no one-to-one relation between 
a stock position~$N_{i,t,j}$ with unrealized gain~$N_{i,t,j}(S^j_t-S^j_i)$ and a taxed amount in the bank account.

\begin{remark} 
Natural generalizations are to introduce dividends or to limit tax payments to a subset of time points in $\{1,\ldots,T\}$. 
We refrain from doing so to avoid complicating the notation, but we stress that the main results of the paper would still hold.
\end{remark}
\begin{definition}
The model satisfies the no-arbitrage (NA) condition iff for every $N\in\mathcal{N}$ the following implication holds:
\beao
V^\alpha(0,N)\ge 0\quad\mbox{a.s.}\quad\implies\quad V^\alpha(0,N)=0\quad\mbox{a.s.}
\eeao
\end{definition}

\subsection{Basic properties of the model}

Next, we establish some basic properties of the tax model. 

\begin{lemma}\label{lemma:concave}
    \begin{enumerate}
    \item Let $x\in\bbr$ and $N\in\mathcal{N}$. The process $\eta$ from (\ref{13.11.2025.3.ck}) that makes $(\eta,N)$ self-financing for the initial capital~$x$ satisfies
    \beam\label{eq:bankalt1}
    \eta_t=x+G_t-\Pi_t-\sum_{i=0}^t\langle N_{i,t},S_i\rangle, \quad t=0,1,\ldots,T.
    \eeam
    \item The terminal wealth is positively homogeneous in $(x,N)$, i.e., for all $\lambda\in\bbr_+$, $x\in \bbr$, and $ N\in \mathcal{N}$, we have $V^\alpha(\lambda x,\lambda N)=\lambda V^\alpha(x,N)$.
    \item For every $x\in \mathbb{R}$ and every $N\in \mathcal{N}$, we have that $V^\alpha(x,N)\le V^0(x,N)$.
    \item The terminal wealth  is concave in $(x,N)$, i.e., for all $\lambda\in[0,1]$, $x^{(1)},x^{(2)}\in \bbr$, and $ N^{(1)},N^{(2)}\in \mathcal{N}$, we have 
\beao
V^\alpha(\lambda x^{(1)}+(1-\lambda) x^{(2)},\lambda N^{(1)}+(1-\lambda)N^{(2)})\ge \lambda V^\alpha(x^{(1)},N^{(1)})+(1-\lambda) V^\alpha(x^{(2)},N^{(2)}),
\eeao
and thus for all $x\in\bbr$, the set~$\{V^\alpha(x,N)\ :\ N\in \mathcal{N}\}-L^0(\F_T;\bbr_+)$ is convex. 
\item The NA condition holds if and only if it holds for the corresponding tax-free model. 
\end{enumerate}
\end{lemma}
\begin{proof}
Ad $(i)$. Follows by considering the increments between $t-1$ and $t$.

Ad $(ii)$. Using (\ref{13.11.2025.3.ck}), it can easily be shown by induction on $t$ that the mapping $(x,N)\mapsto \eta_t(x,N)$, where $\eta$ is given by (\ref{13.11.2025.3.ck}), is positively homogeneous. 

Ad $(iii)$. Let us show by induction that $\eta^{0}_t\ge\eta^{\alpha}_t$ for $t=0,\ldots,T$, where $\eta^0$ and $\eta^\alpha$ are given by the self-financing 
condition with tax rates $0$ and $\alpha$, respectively. The base case $t=0$ is trivial. Suppose that we have shown the assertion for all $s<t$. We get $G^{\alpha}_t\le G^{0}_t$, and the claim immediately follows from (\ref{eq:bankalt1}).

Ad $(iv)$. For $\lambda\in[0,1]$, $x^{(1)},x^{(2)}\in \mathbb{R}$, and $N^{(1)},N^{(2)}\in \mathcal{N}$ define $x:=\lambda x^{(1)}+(1-\lambda)x^{(2)}$ and $N:=\lambda N^{(1)}+(1-\lambda )N^{(2)}$. Let us prove by induction the stronger claim that 
\beam\label{26.12.2025.1}
\eta_t(x,N)\ge \lambda \eta_t(x^{(1)},N^{(1)})+(1-\lambda )\eta_t(x^{(2)},N^{(2)})\quad \mbox{for}\ t=0,\ldots,T.
\eeam
The base case $t=0$ is trivial. Suppose that we have shown (\ref{26.12.2025.1}) for every time~$s<t$.
On the set~$\{\Pi_t(x,N)=\Pi_{t-1}(x,N)\}$, we use the self-financing condition~(\ref{13.11.2025.3.ck}) and (\ref{26.12.2025.1}) for $t-1$ to get
\beao
&&\eta_t(x,N)- \lambda \eta_t(x^{(1)},N^{(1)})-(1-\lambda )\eta_t(x^{(2)},N^{(2)})\\
&\ge&\lambda (\Pi_t(x^{(1)},N^{(1)})-\Pi_{t-1}(x^{(1)},N^{(1)}))+(1-\lambda )(\Pi_t(x^{(2)},N^{(2)})-\Pi_{t-1}(x^{(2)},N^{(2)}))\ \ge\ 0.
\eeao
On the set $\{\Pi_t(x,N)>\Pi_{t-1}(x,N)\}$, we use (\ref{eq:bankalt1}), $\Pi_t(x^{(j)},N^{(j)})\ge \alpha G_t(x^{(j)},N^{(j)})$ for $j=1,2$,  (\ref{26.12.2025.1}) for all $s<t$, and $r\ge 0$ to conclude that
\beao
&&\eta_t(x,N)- \lambda \eta_t(x^{(1)},N^{(1)})-(1-\lambda )\eta_t(x^{(2)},N^{(2)})\\
&=&\sum_{s=1}^{t}r_s(\eta_{s-1}(x,N)- \lambda \eta_{s-1}(x^{(1)},N^{(1)})-(1-\lambda )\eta_{s-1}(x^{(2)},N^{(2)})) - \Pi_t(x,N)\\
& & + \lambda\Pi_t(x^{(1)},N^{(1)})+(1-\lambda)\Pi_t(x^{(2)},N^{(2)})\\
&\ge& \sum_{s=1}^{t}(1-\alpha)r_s(\eta_{s-1}(x,N)- \lambda \eta_{s-1}(x^{(1)},N^{(1)})-(1-\lambda )\eta_{s-1}(x^{(2)},N^{(2)}))\\
&\ge& 0.
\eeao
Ad $(v)$. The ``if'' direction holds by part~(iii). Conversely, assume that the frictionless market does not satisfy NA. Then, it allows for an arbitrage in a single period, that is, there exists a $t\in\{0,\ldots,T-1\}$ and an $N\in\mathcal{N}$ with $N_{i,u}=0$ for all $(i,u)\not= (t,t)$ such that 
$V^0(0,N)=\langle N_{t,t}, S_{t+1} - (1+r_{t+1})S_t\rangle \prod_{s=t+2}^T (1+r_s)$ satisfies $P(V^0(0,N)\ge 0)=1$ and $P(V^0(0,N)>0)>0$.
Since $r_s\ge 0$ and $\alpha<1$, the same holds for 
\beao
V^\alpha(0,N) & = & (1-\alpha)\langle N_{t,t}, S_{t+1} - (1+r_{t+1})S_t\rangle^+ \prod_{s=t+2}^T (1+(1-\alpha)r_s)\\
& & - \langle N_{t,t}, S_{t+1} - (1+r_{t+1})S_t\rangle^- \prod_{s=t+2}^T (1+r_s).
\eeao
\end{proof}

\begin{remark}
The proof of Lemma~\ref{lemma:concave}(iv) shows that for concavity, which is essential for the further analysis, one needs that $r_t\ge 0$.
Otherwise, it may happen that the absence of taxes due to an unavoidable realization of losses is disadvantageous if they have to be paid later.
\end{remark}

Central to the analysis in this paper is a decomposition of a strategy into a reversible part and a ``purely nonreversible part''.
It turns out that it is sufficient to do this decomposition period by period and on a pre-tax basis. 
Going back to Schachermayer~\cite{schachermayer.1992}, for each period, the  space~$\bbr^d$ can be decomposed into a set of null-strategies, which are completely riskless, and an orthogonal complement. Redundant securities lead to a non-trivial set of null-strategies.
Since we have short-selling constraints, we cannot apply orthogonal projections, but we can work with the decomposition in K\"uhn and Molitor~\cite{kuehn.molitor.2019} developed for transaction costs models. The orthogonal complement is replaced by a set of ``purely non-reversible'' trades.

For any $t\in\{0,\ldots,T-1\}$, we define the (convex) cone of \textit{reversible strategies} at time $t$ by
\beam\label{def:reversible}
\mathcal{R}_t:=\{\beta\in L^0(\F_t; \bbr_+^d)\ :\ \langle\beta,S_{t+1}-(1+r_{t+1})S_t\rangle=0 \text{ a.s.}\}.
\eeam
\begin{lemma}[Lemma 3.3 and Lemma 3.4 in K\"uhn and Molitor~\cite{kuehn.molitor.2019}]\label{30.12.2025.1} For every $\beta\in L^0(\F_t;\bbr^{d}_+)$, there exists a unique (up to null sets) decomposition $\beta = p_t(\beta)+q_t(\beta)$ such that $p_t(\beta)\in\mathcal{R}_t$ and $q_t(\beta)\in L^0(\F_t;\bbr^{d}_+)$ satisfies 
$|q_t(\beta)|\le |\beta'|$ for all $\beta'\in L^0(\F_t;\bbr^{d}_+)$ with $\beta-\beta'\in\mathcal{R}_t$. The decomposition is positively homogeneous in the sense that $q_t(\mu\beta) = \mu q_t(\beta)$ for all $\mu\in L^0(\F_t;\bbr_+)$ and $\beta\in L^0(\F_t;\bbr^{d}_+)$, and it
is continuous in the sense that for all $(\beta^n)_{n\in\bbn}\subseteq L^0(\F_t;\bbr^{d}_+)$ and all $\beta\in L^0(\F_t;\bbr^{d}_+)$
\beam\label{14.1.2026.01}
\beta^n\to \beta\ \mbox{a.s. as}\ n\to\infty\quad\implies\quad q_t(\beta^n)\to q_t(\beta)\ \mbox{a.s. as}\ n\to\infty.
\eeam
\end{lemma}
\begin{proof}
By interpreting $\beta$ as ``orders'' in the transaction costs model and observing that $\mathcal{R}_t$ is closed in probability, we can directly apply the lemmas. We note that in the proof of \cite[Lemma~3.4]{kuehn.molitor.2019}, Equation~(3.12) is neither used nor required for the subsequent arguments.
\end{proof}
Now, for each $t\in \{0,\ldots,T-1\}$, we can define the set of \textit{purely nonreversible strategies}
\beam\label{29.12.2025.2}
\mathcal{P}_t:=\{\beta\in L^0(\F_{t};\bbr_+^d)\ :\ q_{t}(\beta)=\beta\},
\eeam
which is closed in probability by (\ref{14.1.2026.01}).

For the convenience of the reader, we recall two key lemmas that are applied in the following. 
\begin{lemma}[Koml\'os' lemma, Lemma A1.1 of \cite{delbaen1994general}]\label{lemma:komlos}
     Let $(f_n)_{n\in \mathbb{N}} \subseteq L^0(\F_t;\bbr_+)$ be a sequence such that the set~${\rm conv}\{f_n\ :\ n\in\bbn\}$ is bounded in $L^0$. Then, there exist a sequence~$(g_n)_{n\in \mathbb{N}}$ with $g_n\in {\rm conv}~\{f_k\ :\ k\ge n\}$ and $g\in L^0(\F_t; \bbr_+)$
such that $g_n\to g$ a.s. as $n\to \infty$.
\end{lemma}

\begin{lemma}[Lemma A.2 of \cite{schachermayer2004fundamental}]\label{Lemma: MeasurableSubsequence}
Let $t\in\{0,\ldots,T\}$. For a sequence~$(f_n)_{n\in\mathbb{N}}\subseteq 
L^0(\F_t;\bbr^{d}_+)$, there is a random subsequence~$(\tau_k)_{k\in\mathbb{N}}$, i.e., 
a strictly increasing sequence of $\mathbb{N}$-valued $\F_t$-measurable random variables such that the sequence of random variables~$(g_k)_{k\in\mathbb{N}}$ given by $g_k(\omega):=f_{\tau_k(\omega)}(\omega)$, $k\in\mathbb{N}$, converges a.s. in the 
one-point-compactification~$\bbr^{d}_+\cup\{\infty\}$ to a random variable $f\in L^0(\F_t;\bbr^{d}_+\cup\{\infty\})$. In fact, we may find the subsequence such that
				\begin{align*}
					| f|=\limsup_{n\to\infty}|f_n|\quad \text{a.s.},
				\end{align*}
				where $|\infty|=\infty$.
\end{lemma}

\begin{lemma}\label{lemma:L0-bound nonreversible} 
For every $t=0,\ldots,T-1$, we define the set
$\mathcal{N}_t:=\{N\in \mathcal{N}\ :\ N_{i,i}=0 \text{ for all}\ i<t\}$ and the 
random variable  
\begin{align*}
    \eps_t:={\rm essinf}\{&\eps1_{A}+\infty1_{A^c}\in L^0(\F_t;[0,\infty])\ :\ \exists \eps\in L^0(\F_t;[0,1])\hspace{1ex}\exists A \in \F_t\hspace{1ex}\exists N\in \mathcal{N}_t\text{ with } \\
    &N_{t,t}\in \mathcal{P}_t \text{ and } |N_{t,t}|=1_A\text{ s.t. } P(V^0(0,N)\le - \eps \mid \F_t)\le \eps\}.
\end{align*}
If the model satisfies NA (cf. Lemma~\ref{lemma:concave}(v)), we have the following results:
    \begin{enumerate}
      \item For every $ t\in \{0,\ldots,T-1\}$ and every $ N \in \mathcal{N}_t $ with $N_{t,t}\in \mathcal{P}_t$, we have
      $P(V^0(0,N)\le - \eps_t|N_{t,t}|\ |\ \F_t)\ge \eps_t$ on ${\{|N_{t,t}|> 0\}}$,
        \item $\eps_t>0$  for $ t=0,\ldots,T-1$,
        \item  For every sequence $(N^n)_n \subseteq\mathcal{N}$ with $\liminf_{n\to \infty} V^0(0,N^n)>-\infty$ a.s., we have\\
        $\limsup_{n\to \infty}|q_t(\sum_{i=0}^tN^n_{i,t})|<\infty$ a.s. for $t=0,\ldots, T-1$.
    \end{enumerate}
    \begin{proof}
Ad $(i)$. Fix $t\in \{0,\ldots,T-1\}$ and assume by contradiction that there exists $N\in \mathcal{N}_t$ with $N_{t,t}\in \mathcal{P}_t$ such that $P(B_N)>0$ for
$B_N:=\{P(V^0(0,N)\le - \eps_t|N_{t,t}|\ |\ \F_t)< \eps_t\}\cap \{|N_{t,t}|>0\}$. 
By homogeneity, we may assume $|N_{t,t}|=1_{\{|N_{t,t}|>0\}}$. 
Observe that the triplet $(1,\{|N_{t,t}|>0\},N)$ satisfies the conditions in the essential infimum.
Thus, we must have 
\beam\label{eq:epsLessInfty}
 B_N\subseteq \{|N_{t,t}|>0\}\subseteq \{\eps_t<\infty\}\quad \mbox{a.s.}
\eeam
Define the random variable 
\beao
K & := & \inf\left\{k\in \mathbb{N}\ :\ 1_{B_N}(P(-\eps_t<V^0(0,N)\le -\eps_t+\frac{1}{k}\ |\ \F_t)+\frac{1}{k})\right.\\
& & \qquad\qquad\qquad \left. \le1_{B_N}\frac{\eps_t-P(V^0(0,N)\le -\eps_t\ |\ \F_t)}{2}\right\}.
\eeao
Since $1_{B_N}(P(-\eps_t<V^0(0,N)\le -\eps_t+1/k\ |\ \F_t)+1/k)$ converges to $0$ a.s. as $k \to \infty$, the infimum is attained (and thus finite) a.s.
We set $\eps_t'=(\eps_t-1/K)1_{B_N}+\infty1_{B_N^c}$. Using that $\{K=k\}\in \F_t$, $k\in \bbn$, we get
\beao
&&1_{B_N}P(V^0(0,N)\le -\eps_t'\ |\ \F_t)\\
&=&1_{B_N}P(V^0(0,N)\le -\eps_t\ |\ \F_t)+1_{B_N}P(-\eps_t<V^0(0,N)\le -\eps_t+\frac{1}{K}\ |\ \F_t)\\
&\le& 1_{B_N}\frac{\eps_t+P(V^0(0,N)\le -\eps_t\ |\ \F_t)}{2}+1_{B_N}\frac{\eps_t-P(V^0(0,N)\le -\eps_t\ |\ \F_t)}{2}-1_{B_N}\frac{1}{K}\\
&=& 1_{B_N}\eps_t'.
\eeao
By the minimality of $\eps_t$, this implies $\eps_t'\ge\eps_t$. Because of (\ref{eq:epsLessInfty}) and $K<\infty$, we obtain a contradiction.

Ad $(ii)$ and $(iii)$.
First observe that for every $N\in \mathcal{N}$ the strategy defined by 
\beam\label{eq:PureStrat}
\wh{N}_{s,s}:=q_s\left(\sum_{i=0}^sN_{i,s}\right) \text{ and } \wh N_{s,s+1}:=0 \text{ for } s=0,\ldots,T-1,
\eeam
satisfies $V^0(x,N)=V^0(x,\wh{N})$ for all $x\in \bbr$.
Let us prove $(ii)$ and the assertion $(iii)$ with $\mathcal{N}_t$ instead of $\mathcal{N}$ by a joint backward induction
on $t=T-1,\ldots,0$. The case $t=T-1$ is a special case of the induction step.\\
Suppose we have already shown that $\eps_s>0$ 
for all $ s=t+1,\ldots,T-1$ and that we have shown $(iii)$ for all sequences $(N^n)_n\subseteq\mathcal{N}_{t+1}$ with $\liminf_{n\to \infty} V^0(0,N^n)>-\infty$. Now let us prove that $\eps_t>0$. It is easy to see that $\eps_t$ is defined over a minimum-stable set. 
Thus, there exists a sequence  
$(\eps_n,A_n,N^n)_{n\in \bbn}\subseteq L^0(\F_t;[0,1])\times\F_t\times\mathcal{N}_t$ with $N^n_{t,t}\in \mathcal{P}_t, |N^n_{t,t}|=1_{A_n}\uparrow 1_{\{\eps_t<\infty\}}$, 
$\eps_n1_{A_n}+\infty1_{A^c_n}\downarrow\eps_t$ a.s.
and $ P(V^0(0,N^n)\le - \eps_n\ |\ \F_t)\le \eps_n$ for all $n\in \bbn$. 
Thus
\beao
P(\{V^0(0,N^n)\le - \eps_n\}\cap B_t\ |\ \F_t)\le \eps_n1_{B_t},\mbox{\ where } B_t:=\{\eps_t=0\}.
\eeao
In particular, by defining $\mu_n:=(V^0(0,N^n1_{B_t}))^+$, we get $V^0(0,N^n1_{B_t})-\mu_n\to 0$ in probability. We transform each $N^n$ according to (\ref{eq:PureStrat}) into $\wh N^n$ and obtain a.s.-convergence by passing to a subsequence.

Since $|\wh N_{t,t}^n1_{B_t}|=|N^n_{t,t}1_{B_t}|=1_{B_t\cap A_n}\uparrow1_{B_t}$, we can apply Lemma~\ref{Lemma: MeasurableSubsequence} to obtain a measurable subsequence $(\tau_k)_k$ such that $\wh{N}^{\tau_k}_{t,t}1_{B_t}$ converges to some $\beta\in L^0(\F_t;\bbr_+^d)$ with $|\beta|=1_{B_t}$. Now define the strategies $N^{(1)}\in \mathcal{N}_t$ via $N^{(1)}_{t,t}=\beta$ and $N^{(1)}_{i,u}=0$ for $(i,u)\neq (t,t)$ and
\beao
{N}^{(2),n}_{i,u}:=\left\{\begin{array}{ll} \wh{N}^n_{i,u} & \text{ if } t+1\le i \le u\\ 0 & \text{ otherwise}.\end{array}\right.
\eeao
Then by linearity of $N\mapsto V^0(0,N)$, we have $V^0(0,{N}^{(2),\tau_k}1_{B_t})-1_{B_t}{\mu}_{\tau_k}\to -V^0(0,N^{(1)}1_{B_t})\text{ a.s.}$
Using the induction hypotheses $(iii)$ on $({N}^{(2),\tau_k}1_{B_t})_k\subseteq\mathcal{N}_{t+1}$ in combination with Lemma~\ref{lemma:komlos} we can deduce the existence of ${N}^{(2)}\in \mathcal{N}_{t+1}$ and ${\mu}\in L^0(\F_T; \bbr_+)$ such that
$V^0(0,{N}^{(2)}1_{B_t})-{\mu}1_{B_t}=-V^0(0,{N}^{(1)}1_{B_t})$.
Since $\mathcal{P}_t$ is closed in probability (Lemma~\ref{30.12.2025.1}) and $|N^{(1)}_{t,t}1_{B_t}|=1$ on $B_t$, we know that $N^{(1)}_{t,t} 1_{B_t}\in \mathcal{P}_t$ and $P(V^0(0,N^{(1)}1_{B_t}) \neq 0\ |\ \F_t)>0$ on $B_t$. Now assume by contradiction that $P(B_t)>0$. Then, $P(V^0(0,N^{(1)}1_{B_t})\neq 0)>0$. 
If $P(V^0(0,N^{(1)}1_{B_t}) < 0)=0$, then $N^{(1)}1_{B_t}$ is an arbitrage; otherwise, if $P(V^0(0,N^{(1)}1_{B_t}) < 0)>0$, then $N^{(2)}1_{\{V^0(0,N^{(1)}1_{B_t}) < 0\}}$ is an arbitrage, since $\{V^0(0,N^{(1)}1_{B_t})<0\}\in \F_{t+1}$. 
This contradicts NA, and thus $\{\eps_t=0\}$ must be a null set.

Under the induction hypothesis and given that $(ii)$ holds for $t$, we now establish $(iii)$ for all sequences $(N^n)_n\subseteq\mathcal{N}_t$ with $\liminf_{n\to \infty} V^0(0,N^n)>-\infty$ a.s. Let $(N^n)_n\subseteq\mathcal{N}_t$ and define for each $n$ the corresponding strategy $\wh{N}^n$ given by the transformation (\ref{eq:PureStrat}).
Define the set $A_t:=\{\limsup_{n\to \infty}|\wh{N}_{t,t}^n| = \infty\}$
and suppose that $P(A_t)>0$.
Define a measurable subsequence by $\tau_0:=0$ and $\tau_k:=\min\{n > \tau_{k-1}:1_{A_t}|\wh{N}^n_{t,t}|\ge1_{A_t}k\}$. 
Now observe that by $\eps_t>0$, the sequence
$(\inf_{n\ge k}V^0(0,\wh{N}^{\tau_{n}})+\eps_t|\wh{N}_{t,t}^{\tau_{k}}|)_{k\in \mathbb{N}}$ converges a.s. to $\infty$ on $A_t$ as $k\to \infty$, which implies
\beam\label{eq:KeyLemmaClosedContra}
    \lim_{k\to \infty}P(\{V^0(0,\wh{N}^{\tau_{k}})+\eps_t|\wh{N}^{\tau_k}_{t,t}|\le 0\}\cap A_t)= 0.
\eeam
On the other hand, our assumption $P(A_t)>0$ together with $(i)$ and $(ii)$ yield 
\beao
\E[P(\{V^0(0,\wh{N}^{\tau_k}) +\eps_t|\wh{N}_{t,t}^{\tau_k}|\le0\}\cap A_t\ |\ \F_t)]\ge \E[\eps_t1_{A_t}]>0 \quad\forall k\in \mathbb{N}
\eeao
which contradicts (\ref{eq:KeyLemmaClosedContra}). 
Therefore,
$\limsup_{n\to\infty} |\wh{N}_{t,t}^n|<\infty$  a.s. 
Now consider the sequence $\wh{N}^{(2),n}\subseteq\mathcal{N}_{t+1}$ defined by
\beao
\wh{N}^{(2),n}_{i,u}:=\left\{\begin{array}{ll} \wh{N}^n_{i,u}&\text{for } 1\le t+1\le i\le u \\ 0 &\text{otherwise}.\end{array}\right.
\eeao
Since $\limsup_{n\to\infty} |\widehat{N}_{t,t}^n|<\infty$ and $\liminf_{n\to \infty} V^0(0,\widehat{N}^n)>-\infty$ a.s., the linearity of the mapping $N\mapsto V^0(0,N)$ implies that $\liminf_{n\to \infty} V^0(0,\widehat{N}^{(2),n})>-\infty$ a.s. as well. Applying the induction hypothesis to $(\wh N^{(2),n})_{n\in \bbn}$ shows that $\limsup_{n\to\infty}|\wh N^n_{i,i}|<\infty$ for $i=t,t+1,\ldots,T-1$ which completes the proof.
\end{proof}
\end{lemma}

\section{Utility maximization}\label{29.12.2025.3}

We analyze the utility maximization problem under the natural constraint that the terminal wealth must be nonnegative.
Apart from this, we impose only minimal assumptions on the utility function: 
\begin{assumption}\label{14.2.2026.2} 
Let $U: \bbr\to\bbr\cup\{-\infty\}$ be a function that takes the value~$-\infty$ on $\bbr_-\setminus\{0\}$, whose restriction to $\bbr_+\setminus\{0\}$ is $\bbr$-valued, nondecreasing, and concave, and that satisfies $U(0)=\lim_{x\searrow 0}U(x)$.
\end{assumption}
We introduce the set
\beam\label{17.9.2026.1}
{\mathcal{V}^\alpha(x)}:=\{V^\alpha(x,N)\ :\ N\in \mathcal{N}\}-L^0(\F_T;\bbr_+)
\eeam 
of attainable terminal wealth with initial capital~$x\in\bbr$. In addition, we define
\beao
\mathcal{A}^\alpha(x):=\{N\in \mathcal{N}\ :\ V^\alpha(x,N)\ge 0 \text{ a.s.}\}
\eeao
and 
\beao
\mathcal{V}_{\ge0}^\alpha(x):=\mathcal{V}^\alpha(x)\cap L^0(\F_T;\bbr_+)
=\left(\{V^\alpha(x,N)\ :\ N\in \mathcal{A}^\alpha(x)\}-L^0(\F_T;\bbr_+)\right)\cap L^0(\F_T;\bbr_+).
\eeao
For $x\in\bbr_+\setminus\{0\}$, $N\in \mathcal{N}$ we define the expected utility as 
\beao
\E[U(V^\alpha(x,N))]:=\E[U^+(V^\alpha(x,N))]-\E[U^-(V^\alpha(x,N))]
\eeao
with the convention $\E[U(V^\alpha(x,N))]:=-\infty$ if $\E[U^-(V^\alpha(x,N))]=\infty$, and consider the maximization problem
\beam\label{eq:sup utility}
    u^\alpha(x):=\sup_{N\in\mathcal{A}^\alpha(x)}\E[U(V^\alpha(x,N))].
\eeam
In frictionless markets, that is, $\alpha=0$, it is a well-known result that $u^0(x)< \infty$ for some $x>0$ implies that $u^0(x)<\infty$ for all $x>0$. We extend this result to the tax rate~$\alpha<1$.
\begin{proposition}\label{lemma:IntegrabilityAllx}
If there exist $\alpha_0\in [0,1)$ and $x_0>0$ such that $u^{\alpha_0}(x_0)<\infty$, then
\beao
u^\alpha(x)<\infty \quad\forall x >0 \quad\forall\alpha\in [0,1).
\eeao
\begin{proof}
We only have to show that for a fixed $\alpha\in (0,1)$ and a fixed $x>0$, $u^\alpha(x)< \infty$ implies that $u^0(x)< \infty$. The rest follows from
R\'asonyi and Stettner~\cite[Remark~1.1]{rasonyi.stettner.2006} and Lemma~\ref{lemma:concave}(iii).

{\em Step 1:} For $N\in \mathcal{A}^0(x)$ define the process $\xi_t:=\sum_{s=1}^t1_{\{V^0_{s}(x,N)\ge 0\}}$. Let us show that: 
\beam\label{11.8.2025.2.ck}
V^\alpha(x,\wt{N}) \ge (1-\alpha)^{\xi_T} V^0(x,N),\quad\mbox{where}\ \wt{N}_{t,t}:=(1-\alpha)^{\xi_t} \sum_{i=0}^t N_{i,t},\ \wt{N}_{t,t+1}:=0
\eeam
for $t=0,\ldots,T-1$. Since $\wt{N}$ realizes gains and losses immediately, it is possible to define a meaningful one-dimensional wealth {\em process}~$\ov{V}$ of $\wt{N}$ by the recursion $\ov{V}_0=x$ and
\beao
& & \ov{V}_t = (1+r_t)\ov{V}_{t-1} +\langle \wt{N}_{t-1,t-1},S_t-(1+r_t)S_{t-1}\rangle\\
& & \qquad\ - \alpha\left[(1+r_t)\ov{V}_{t-1} + \langle\wt{N}_{t-1,t-1},S_t-(1+r_t)S_{t-1}\rangle-\ov{V}^\star_{t-1}\right]^+,\quad t=1,\ldots,T,
\eeao
where $\ov{V}^\star_{t-1}:=\max_{s=0,\ldots,t-1}\ov{V}_s$.  
We note that by $\alpha<1$, after the payment of taxes the wealth still attains its running maximum.
On the other hand, we define the fictitious wealth process~$\un{V}$ of $\wt{N}$ that would occur with a wealth tax due at the end of each period---but only when the wealth is nonnegative---by $\un{V}_0=x$ and 
\beao
& & \un{V}_t = (1+r_t)\un{V}_{t-1} +\langle \wt{N}_{t-1,t-1},S_t-(1+r_t)S_{t-1}\rangle\\
& & \qquad\ - \alpha\left[(1+r_t)\un{V}_{t-1} + \langle\wt{N}_{t-1,t-1},S_t-(1+r_t)S_{t-1}\rangle\right]^+,\quad t=1,\ldots,T.
\eeao
Using that $\ov{V}^\star_{t-1}\ge x>0$ and $\alpha<1$, one shows by induction on $t$ that $\ov{V}_t\ge \un{V}_t$ for $t=0,\ldots,T$.
Then, let us show by induction on $t$ that $\un{V}_t=(1-\alpha)^{\xi_t} V^0_t(x,N)$ for $t=0,\ldots,T$. Assume that the assertion is already proven for $t-1$. 
We use that $V^0_t(x,N)$ satisfies the recursion given in (\ref{eq:frictionlesswealth}) and make first the observation that $(1+r_t)\un{V}_{t-1}(\omega) + \langle\wt{N}_{t-1,t-1},S_t-(1+r_t)S_{t-1}\rangle(\omega) \ge 0$ iff $V^0_t(x,N)(\omega)\ge 0$, which then implies the assertion for $t$. Putting this together, we arrive at (\ref{11.8.2025.2.ck}). 

{\em Step 2:} Now, we can again apply the above mentioned argument by R\'asonyi and Stettner~\cite[Remark~1.1]{rasonyi.stettner.2006},
which we briefly repeat for the convenience of the reader. Suppose that $u^\alpha(x)< \infty$ and assume by contradiction that $u^0(x)=\infty$.  The latter means that there exists a sequence~$(N^n)_{n\in\bbn}\subseteq\mathcal{A}^0(x)$ with $\E[U^-(V^0(x,N^n))]<\infty$ for all $n\in\bbn$ and 
\beam\label{11.8.2025.1.ck}
\E[U(V^0(x,N^n))] \to \infty\quad \mbox{as}\ n\to\infty.
\eeam
One considers the strategies $N^n/2$, which satisfy $(1-\alpha)^T V^0(x,N^n/2)=(1-\alpha)^T/2(V^0(x,N^n)+x\prod_{t=1}^T(1+r_t))$. Since $U$ is concave, we have that
\beao
& & U((1-\alpha)^T V^0(x,N^n/2))\nonumber\\
& & \ge \frac12(1-\alpha)^T U(V^0(x,N^n)) + \left(1-\frac12(1-\alpha)^T\right)U\left(\frac{(1-\alpha)^T}{2-(1-\alpha)^T}x\prod_{t=1}^T(1+r_t)\right). 
\eeao
Since 
\beao
U\left(\frac{(1-\alpha)^T}{2-(1-\alpha)^T}x\prod_{t=1}^T(1+r_t)\right) \ge U\left(\frac{(1-\alpha)^T}{2-(1-\alpha)^T}x\right) >-\infty, 
\eeao
(\ref{11.8.2025.1.ck}) implies that 
\beao
\E[U((1-\alpha)^T V^0(x,N^n/2))]\to \infty\quad \mbox{as}\ n\to\infty.
\eeao
Because of inequality~(\ref{11.8.2025.2.ck}) applied to $N:=N^n/2$, this would only be possible if $u^\alpha(x)=\infty$.
\end{proof}
\end{proposition}

The following theorem is the first main result of the paper.

\begin{theorem}\label{Theo:MainUtility} Let $\alpha \in [0,1)$ and $x>0$. Suppose the market satisfies the NA condition, and that the set~$\mathcal{V}_{\ge 0}^\alpha(x)$ is closed in probability. If $u^\alpha(x)<\infty$,
then there exists an optimal strategy~$N\in \mathcal{A}^\alpha(x)$, that is, 
\beao
\E[U(V^\alpha(x,N))]=u^\alpha(x).
\eeao
\end{theorem}

In contrast to frictionless markets, NA does not imply that the set of attainable terminal wealth is closed in probability. This is shown in the following example, in which the supremum in a log-utility maximization problem is {\em not} attained. Sufficient conditions for the closedness are provided in Subsection~\ref{22.8.2026.1.ck}. The basic idea of the example, already described in the introduction, is that in the first period the stock has the same return as the bank account for sure, but
a debt-financed purchase at time~$0$ provides a tax advantage if the stock return between $1$ and $2$ is higher than the interest rate
and a part of the stocks are sold at time~$2$. 

\begin{example}[$\nexists$ optimal strategy]\label{example:Nonclosedness}
Let $T=3$, $\alpha\in(0,1)$, and $r\in\bbr_+\setminus\{0\}$ be the constant interest rate. We introduce the random variables $Y$, $Z_2$, and $Z_3$ with
(conditional) probabilities $P(Y=k)=2^{2-k}$, $P(Z_2=c_k\ |\ Y=k)=p_{1,k}$, $P(Z_2=-r/k\ |\ Y=k)=1-p_{1,k}$,
$P(Z_3=c_k\ |\ Y=k, Z_2=c_k)=p_{2,k}$, $P(Z_3=-r/k\ |\ Y=k, Z_2=c_k)=1-p_{2,k}$, and
$P(Z_3=0\ |\ Y=k, Z_2=-r/k)=1$ for all $k\in \bbn$ with $k\ge 3$, where $c_k:=(k-2)(2r+r^2)/\left(k(1+r)\right)$. The probabilities~$p_{1,k},p_{2,k}\in(0,1)$ for $k\ge 3$ are not yet specified. The filtration is given by $\mathcal{F}_0=\{\emptyset,\Omega\}$, $\mathcal{F}_1=\sigma(Y)$, $\mathcal{F}_2=\sigma(Y,Z_2)$, and $\mathcal{F}_3=\sigma(Y,Z_2,Z_3)$.

We have a single stock with price process~$(S_t)_{t=0,1,2,3}$ given by $S_0=1$, $S_1=1+r$, $S_2=(1+r)(1+r+Z_2)$, and $S_3=(1+r)(1+r+Z_2)(1+r+Z_3)$. The model obviously satisfies NA. Consider the utility maximization problem
\beam\label{eq:OPT1}
u(1):=\sup_{N\in\mathcal{N}}\E[\ln(V^\alpha(1,N))]
\eeam
with the convention that $\ln(x):=-\infty$ for $x\le 0$ and $\E[\ln(V^\alpha(1,N))]:=-\infty$ if $\E[\ln^-(V^\alpha(1,N))]=\infty$. In addition, we consider the auxiliary frictionless maximization problem
\beam\label{eq:OPT2}
\wt u(1):=\sup_{\wt N\in\wt{\mathcal{N}}}\E[\ln((1-\alpha)V^0(1,\wt N)+\alpha)],
\eeam
where $\wt{\mathcal{N}}$ is the set of strategies with regard to the larger filtration given by $\wt \F_0=\wt\F_1=\F_1$, $\wt\F_2=\F_2$, and $\wt\F_3=\F_3$.

Now, we choose the probabilities $p_{1,k}$ and $p_{2,k}$ such that the set of maximizers of the frictionless problem~(\ref{eq:OPT2}) is given by $\{\wt{N}\in\wt{\mathcal{N}}\ :\ \wt{N}_{0,1}+\wt{N}_{1,1}=Y\ \mbox{and}\ (\wt{N}_{0,2}+\wt{N}_{1,2}+\wt{N}_{2,2})1_{\{Z_2>0\}}=Y/21_{\{Z_2>0\}}\}=:\wt{\mathcal{M}}\not=\emptyset$.
This means that the amounts invested in the stock during the second and third period on $\{Z_2>0\}$ are uniquely determined by optimality, whereas the amount invested during the first period and during the third period on $\{Z_2<0\}$ is arbitrary.
To see that this is possible, we first observe that $V^0(1,\wt{N})$ is optimal for (\ref{eq:OPT2}) if and only if it is optimal for $\E(\wt{U}(V^0(1,\wt{N})/(1+r)^3))$ with $\wt{U}(y):=\ln((1-\alpha)y+\alpha/(1+r)^3)$, $y\in\bbr$. 
For strategies from $\wt{\mathcal{M}}$, the frictionless wealth is nondecreasing for sure, and thus 
$\wt{U}(V^0(1,\wt{N})/(1+r)^3)\ge \ln(1/(1+r)^3)>-\infty$.
In addition, $V^0(1,\wt{N})/(1+r)^3\le C(1+Y)$ for some $C\in\bbr_+$. Conditionally on $\{Y=k\}$, we have a complete market with a unique equivalent martingale measure~$Q$. The optimality condition that the marginal utility is proportional to $dQ/dP$
(see, e.g., \cite[Theorem~2.0]{kramkov.schachermayer.1999}) leads to $3$ equations with $3$ unknowns that have a unique solution. 
The utility function $\wt{U}$ satisfies the Inada condition at the left endpoint of the interval on which it is finite, rather than at zero,
but the theorem is still applicable.
Without the first period and the third period on $\{Z_2<0\}$, the market is non-redundant, and the $2$-dimensional maximizer~$(\wt{N}_{0,1}+\wt{N}_{1,1},(\wt{N}_{0,2}+\wt{N}_{1,2}+\wt{N}_{2,2})1_{\{Z_2>0\}})$ is unique.

Next observe that by $\mathcal{N}\subseteq \wt{\mathcal{N}}$ and by
\beam\label{8.2.2026.1}
V^\alpha(1,N)=1+G_3^\alpha-\Pi_3\le (1-\alpha)\left[1+G_3^\alpha\right]+\alpha\le (1-\alpha)V^0(1,N)+\alpha\quad\mbox{for all}\ N\in\mathcal{N},
\eeam
we have $\wt u(1)\ge u(1)$. To show equality, we consider the sequence $(N^n)_{n\in \bbn}\subseteq \mathcal{N}$ given by $N^n_{0,0}:=n$, $N^n_{0,1}:=Y1_{\{Y\le n\}}$, $N^n_{0,2}:=Y/21_{\{Y\le n, Z_2>0\}}+Y1_{\{Y\le n, Z_2<0\}}$,  $N^n_{1,1}:=0$, and $N^n_{2,2}:=0$.
By the choice of $c_k$, we have that
\beam\label{28.12.2025.1}
(k-1)(2r+r^2) = \frac{k}2 \left(2r + c_k + r^2 + rc_k\right)\quad\mbox{for all } k\ge 3,
\eeam
which guarantees that the strategies~$N^n$ satisfy $G_2(N^n)=0$, $G_3(N^n)>0$ on $\{Y\le n, Z_2>0\}$ and $G_2(N^n)\le0$, $G_3(N^n)>0$ on $\{Y\le n, Z_2<0\}$. Hence, we obtain
\beam\label{9.2.2026.1}
V^\alpha(1,N^n)=(1-\alpha)V^0(1,N^n)+\alpha\quad\mbox{on}\ \{Y\le n\}.
\eeam
Since $0\le\ln(V^\alpha(1,N^n))\le \ln(\wt C(1+Y))$ for all $n\in\bbn$ for some $\wt C\in\bbr_+$, (\ref{8.2.2026.1}) and (\ref{9.2.2026.1}) imply that
\beao
\E[\ln(V^\alpha(1,N^n))]\to \wt{u}(1)\quad\mbox{as}\ n\to\infty,
\eeao
which yields $u(1)=\wt u(1)$. Together with inequality~(\ref{8.2.2026.1}) and the fact that the set of maximizers of (\ref{eq:OPT2}) is given
by $\wt{\mathcal{M}}$, this implies that any maximizer of (\ref{eq:OPT1}) must satisfy the following conditions:
\begin{enumerate}
    \item[(i)] $N_{0,1}+N_{1,1}=Y$, $(N_{0,2}+N_{1,2}+N_{2,2})1_{\{Z_2>0\}}=Y/21_{\{Z_2>0\}}$,
    \item[(ii)] $V^\alpha(1,N)=(1-\alpha)V^0(1,N)+\alpha$.
\end{enumerate}
Let us show that no such $N\in\mathcal{N}$ exists. Assume that $N\in\mathcal{N}$ satisfies (i). W.l.o.g. we can assume that $N_{0,0}\ge 1$, that is, the initial capital is invested in the stock and not in the bank account with the same return.
Since $N_{0,0}$ is non-random, there exists $k\in\bbn$ with $k>2N_{0,0}$, and one has $P(Y=k, Z_2>0)>0$.
On the set~$\{Y=k, Z_2>0\}$, the loss pool immediately before stocks are sold at time~$2$ is $(k-1)(2r+r^2) - r N_{1,1}$, which coincides with $k/2(1+r)(r+c_k)-r(N_{1,1}-k/2)$ by (\ref{28.12.2025.1}).
On the other hand, on $\{Y=k, Z_2>0\}$, one has that $N_{1,1}>k/2$, and $k/2$ stocks are sold at time~$2$. Thus, even if one takes the stocks purchased at time~$1$, which have the lower book profit $(1+r)(r+c_k)$ per share, taxes must be paid prematurely. Then, it follows by $r_3=r>0$ that $P((1-\alpha)V^0(1,N)+\alpha>V^\alpha(1,N))>0$. Thus, $N$ does not satisfy (ii), and the supremum in (\ref{eq:OPT1}) is not attained.

Furthermore, for an arbitrary $\wt{N}\in\wt{\mathcal{M}}$, one has $V^\alpha(1,N^n)\to (1-\alpha)V^0(1,\wt{N})+\alpha$ in probability as $n\to\infty$, which shows that $\{V^\alpha(1,N) : N\in\mathcal{N}\} - L^0(\mathcal{F}_T; \bbr_+)$ is not closed in probability.
\end{example}
We now consider the case without taxes. 
\begin{corollary}[Lemma~3.1 in Napp~\cite{napp.2003} and Corollary of Theorem~\ref{Theo:MainUtility}]\label{6.1.2026.1.ck} Let $\alpha=0$ and $x\in\bbr$. 
Suppose that the market, which can be identified with a frictionless market with short-selling constraints, satisfies NA. Then, the set~$\{V^0(x,N)\ :\ N\in\mathcal{N}\}-L^0(\F_T; \bbr_+)$ is closed in probability. If, in addition,
$x>0$ and $u^0(x)<\infty$, then there exists an optimal strategy.
\end{corollary}
The first assertion of Corollary~\ref{6.1.2026.1.ck} has already been proven by Napp~\cite{napp.2003} for general polyhedral cone constraints.
Her inspiring inductive proof reduces the dimension of the trading strategy, using projections along null-strategies. Applied to the constraint set~$\bbr^d_+$, strategies in the lower-dimensional constraint set can have negative components, but the reduced constraint set is again a polyhedral cone, which is why the induction argument works. 
By contrast, our proof is based on the more general cone-decomposition of K\"uhn and Molitor~\cite{kuehn.molitor.2019} (cf. Lemma~\ref{30.12.2025.1}).

For earlier work on arbitrage under short-selling and more general portfolio constraints, see Jouini and Kallal~\cite{jouini.kallal.1995}, Pham and 
Touzi~\cite{pham.touzi.1999}, and Carassus, Pham and Touzi~\cite{carassus.pham.touzi.2001}. More recent related results can be found in Fontana and Runggaldier~\cite{fontana.runggaldier.2021}.

In continuous time, there is the related result of Czichowsky and Schweizer~\cite[Corollary 4.7]{czichowsky.schweizer.2011} that the set of stochastic integrals with integrands satisfying predictable polyhedral convex constraints is closed with regard to the semimartingale topology.
This is used by Pulido~\cite{pulido.2014} to derive a FTAP for nonnegative, locally bounded asset price processes under short-selling constraints.
However, Pulido's FTAP is formulated in terms of NFLVR, whereas the first assertion of Corollary~\ref{6.1.2026.1.ck} assumes only NA and does not follow directly from these results.

The second assertion of Corollary~\ref{6.1.2026.1.ck} generalizes R\'asonyi and Stettner~\cite[Theorem~1.1]{rasonyi.stettner.2006} to markets with short-selling constraints.
\begin{proof}[Proof of Corollary~\ref{6.1.2026.1.ck}]
Under NA, closedness follows from (\ref{eq:PureStrat}), Lemma~\ref{lemma:L0-bound nonreversible}(iii), and Lemma~\ref{lemma:komlos}.
Thus, for every $x>0$, the set~$\mathcal{V}^0(x)$ is closed in probability, and hence so is $\mathcal{V}_{\ge0}^0(x)$. If $u^0(x)<\infty$, an optimizer exists by Theorem~\ref{Theo:MainUtility}.
\end{proof}

\begin{lemma}[Key Lemma]\label{lemma:KeyUtility}
Assume that the corresponding tax-free market satisfies NA and $u^0(x)<\infty$ for some $x>0$. Then, for every $x>0$ there exist $\wh{x}>0$ and $\wh{N}\in \mathcal{A}^0(\wh{x})$ such that
\beao
\E[U^+(V^0(\wh{x},\wh{N}))]< \infty\quad\mbox{and}\quad V^0(x,N)\le V^0(\wh{x},\wh{N})\quad\forall N \in \mathcal{A}^0(x).
\eeao 
\end{lemma}
It is easy to see that the existence of such a dominating terminal wealth directly leads to the desired result:
\begin{proof}[Proof of Theorem~\ref{Theo:MainUtility} given that Lemma~\ref{lemma:KeyUtility} holds]
First, observe that the assumptions NA and $u^\alpha(x)<\infty$ for some $x>0$ do not depend on the choice of $\alpha\in[0,1)$ (Lemma~\ref{lemma:concave}(v) and Proposition~\ref{lemma:IntegrabilityAllx}). We find a sequence~$(N^n)_{n\in\bbn}\subseteq\mathcal{A}^\alpha(x)$ such that 
\beao
\lim_{n\to \infty}\E[U(V^{\alpha}(x,N^n))]= u^\alpha(x).
\eeao
From Lemma~\ref{lemma:concave}(iii) and Lemma~\ref{lemma:KeyUtility} we know that 
\beao
0\le\sup_{n \in \mathbb{N}}V^{\alpha}(x,N^n)\le \sup_{n \in \mathbb{N}}V^{0}(x,N^n)\le V^0(\wh{x},\wh{N})\quad \mbox{for some } \wh x>0\mbox{ and } \wh N\in \mathcal{A}^0(\wh x).
\eeao
With Lemma~\ref{lemma:komlos}, we get a sequence $(\wt V^\alpha_n)_{n\in \bbn}$
with 
\beao
\wt{V}_n^\alpha\in \mbox{conv}\{V^\alpha(x,N^k)\ :\ k\ge n\},\ n\in\bbn,\quad \mbox{such that }\wt{V}_n^\alpha \to \wt{V}^\alpha \quad\mbox{a.s. as } n \to \infty
\eeao
for some $\wt{V}^\alpha\in L^0(\F_T; \bbr_+)$. By concavity of $U$, it follows that $\limsup_{n \to \infty}\E[U(\wt{V}_n^\alpha)] \ge u^\alpha(x)$. 
Since $\mathcal{V}_{\ge 0}^\alpha(x)$ is convex, we obtain $(\wt{V}^\alpha_n)_{n\in\bbn}\subseteq\mathcal{V}_{\ge 0}^\alpha(x)$, that is, for all $ n \in \mathbb{N}$, there exist $\wt{N}^n\in \mathcal{A}^\alpha(x)$ and $\wt{\mu}_n\in L^0(\F_T;\bbr_+)$ such that $\wt{V}^\alpha_n=V^\alpha(x,\wt{N}^n)-\wt{\mu}_n$. 
 Since $\mathcal{V}_{\ge 0}^\alpha(x)$ is
closed in probability by assumption, there is an $\wt{N}\in \mathcal{A}^\alpha(x)$ such that $V^\alpha(x,\wt{N})\ge \wt{V}^\alpha$.
Since $U$ is nondecreasing we get $$\sup_{n\in \bbn}U(\wt{V}_n^\alpha)\le \sup_{n\in \bbn}U(V^\alpha(x,\wt{N}^n))\le U^+(V^0(\wh{x},\wh{N})).$$ We now use the lemma of Fatou and the continuity of $U$ to arrive at $$u^\alpha(x)\le\limsup_{n\to \infty}\E[U(\wt{V}_n^\alpha)]\le \E[\limsup_{n\to \infty}U(\wt{V}_n^\alpha)] =\E[U(\wt{V}^\alpha)]\le \E[U(V^\alpha(x,\wt{N}))].$$
\end{proof}
\subsection{Proof of Lemma~\ref{lemma:KeyUtility}}

From now on we work with the corresponding frictionless market.
For a strategy $N\in \mathcal{A}^0(x)$ we introduce the normalized one-period strategy  defined
by
\beam\label{eq:normalizedStrat}
    \beta_{t}(N):=\frac{1}{V^0_{t}(x,N)}q_t\left(\sum_{s=0}^{t}N_{s,t}\right), \qquad t=0,\ldots,T-1,
\eeam
with the convention $\beta_t(N):=0$ on $\{V^0_{t}(x,N)=0\}$. Under NA, we have $P(V^0_{t}(x,N)<0)=0$ and $\{V^0_{t}(x,N)=0\}\subseteq \{V^0_T(x,N)=0\}$ a.s., and the definition allows us to rewrite the frictionless wealth process from (\ref{eq:frictionlesswealth}) as a product:
\beam\label{eq:WealthProduct}
    V_{t}^0(x,N)=x\prod_{s=1}^{t}(1+r_{s}+\langle\beta_{s-1}(N),S_{s}-(1+r_{s})S_{s-1}\rangle)
\eeam
with (a.s.) nonnegative factors.\\ 

Before we can begin our proof we need to introduce random sets.
\begin{definition}\label{3.9.2026.1}
    Fix $t\in \{0,\ldots,T-1\}$ and consider the probability space $(\Omega,\F_t,\PM)$.   
A set-valued mapping $M:\Omega\rightrightarrows \bbr^n$ is called an $\F_t$\textit{-measurable random set} if for every open set $O \subseteq\bbr^n$ we have 
\beao
M^{-1}(O):=\{\omega\in \Omega\ :\ M(\omega)\cap O\neq \emptyset\}\in \F_t.
\eeao
Furthermore, an $\F_t$-measurable random set $M$ is called\textit{ closed-valued} (resp. \textit{compact-valued}) if for every $\omega\in \Omega$ the 
set~$M(\omega)$ is closed (resp. compact) with regard to the Euclidean norm in $\bbr^n$. 
We denote by $L^0(\F_t; M)$ the space of equivalence classes of $\F_t$-measurable random vectors taking values in $M$ a.s.
\end{definition}

\begin{note}
Let $M$ be a non-empty closed-valued $\mathcal{F}_t$-measurable random set. For each $\omega\in\Omega$, the dimension of $M(\omega)\subseteq \bbr^n$ is defined as the maximal number of linearly independent elements of $M(\omega)$.
Then, the mapping~${\rm dim}(M)$ that sends each $\omega\in \Omega$ to the dimension of $M(\omega)$ is $\mathcal{F}_t$-measurable.
\end{note}
\begin{proof}
Closed-valued random sets have a Castaing representation, meaning that there exists a sequence~$(X_n)_{n\in\bbn}$ of $\F_t$-measurable $\bbr^n$-valued random variables such that $M(\omega)=\ov{\{X_n(\omega)\ :\ n\in \bbn\}}$ for all $\omega\in \Omega$ (see, e.g., Pennanen and Perkki\"o~\cite[Theorem~1.27]{pennanen.perkkioe.2024}). Define $X_{\infty}:=0$, $\tau_1(\omega):=\inf\{n\in \bbn\ :\ X_n(\omega)\neq 0\}$ and $\tau_k(\omega):=\inf\{n\in \bbn\ :\ X_n(\omega)\not\in{\rm span}\{X_{\tau_1}(\omega),\ldots,X_{\tau_{k-1}}(\omega)\}\}$. We get $\dim(M)(\omega)=\dim(\{X_n(\omega)\ :\ n\in \bbn\})=\sum_{k=1}^n1_{\{\tau_k<\infty\}}(\omega)$.
\end{proof}

\begin{lemma}\label{lemma:CompactSetA}
For every $t=0,\ldots,T-1$, we define the set 
\beao
A_t:=\{\beta_t\in \mathcal{P}_{t}\ :\ 1+r_{t+1}+\langle\beta_t,S_{t+1}-(1+r_{t+1})S_{t}\rangle\ge 0 \text{ a.s.}\}.
\eeao
Under NA, we have that there exists a compact-valued $\F_{t}$-measurable random set~$M_{t}:\Omega\rightrightarrows \bbr^d$ such that 
$A_t=L^0(\F_t; M_{t})$.
\end{lemma}
\begin{proof}
Since $\mathcal{P}_t$ is closed in probability by (\ref{14.1.2026.01}), we get that $A_t$ is closed in probability. For all $\beta^1_t,\beta^2_t\in A_t$ and for all $ M\in \F_{t}$, we know by Lemma~\ref{30.12.2025.1} that $1_M\beta^1_t+1_{M^c}\beta^2_t\in \mathcal{P}_t$, and thus $1_M\beta^1_t+1_{M^c}\beta^2_t\in A_t$. By \cite[Theorem 1.60]{pennanen.perkkioe.2024}, we get the existence of a closed-valued 
$\F_{t}$-measurable random set~$M'_{t}$ such that $A_t=L^0(\F_t; M'_{t})$. We use $M'_{t}$ to construct a compact-valued random set. For this, 
let $\beta_t\in A_t$ and assume $P(|\beta_t|> 0)>0$. The random variable~$\langle\beta_t,S_{t+1}-(1+r_{t+1})S_{t}\rangle\prod_{s=t+2}^T(1+r_s)$ is an attainable terminal wealth with initial capital~$0$ and a strategy from $\mathcal{N}_t$.
Thus, by Lemma~\ref{lemma:L0-bound nonreversible}(i) and the definition of $A_t$, we obtain
\beam\label{19.9.2026.1}
P(\prod_{s=t+1}^{T}(1+r_s)\ge\varepsilon_t|\beta_t|\  |\ \F_t)\ge\varepsilon_t
\quad\text{on }\{|\beta_t|>0\}\subseteq \{\varepsilon_t<\infty\}.
\eeam
Define the $\F_t$-measurable random variable~$\xi:=\inf\{n\in\mathbb N\ :\ P(\prod_{s=t+1}^{T}(1+r_s)\ge n\varepsilon_t\ |\ \F_t)<\varepsilon_t\}$.
By Lemma~\ref{lemma:L0-bound nonreversible}(ii), $\xi$ is finite a.s. 
Using $\{\xi=n\}\in\mathcal F_t$, $n\in\mathbb N$, estimate~(\ref{19.9.2026.1}) and the definition of $\xi$ imply that $P(\{\xi=n\}\cap\{|\beta_t|\ge n\})=0$ and thus
$|\beta_t|<\xi$ a.s.
The closed ball~$B(0, \xi)$ centered at $0$ with random radius~$\xi$ and $M_t:=M'_t\cap B(0, \xi)$ are compact-valued $\F_t$-measurable random sets.
In addition, we have that $A_t\subseteq L^0(\F_t; B(0, \xi))$. It follows that $L^0(\F_t; M'_t)=A_t\subseteq L^0(\F_t; B(0,\xi))\cap L^0(\F_t; M'_t)=L^0(\F_t; B(0,\xi)\cap M'_t)$ and thus $A_t=L^0(\F_t; M_t)$. 
\end{proof}

\begin{lemma}\label{lemma:BaseForCompSet}
For any compact-valued $\F_t$-measurable random set~$K:\Omega\rightrightarrows \bbr^d$ with $0\in K$ and $\dim(K)\le k$ a.s. we can find $\F_t$-measurable
random variables~$(Y^i)_{i=1,\ldots,k}$ with $Y^i(\omega)\in K(\omega)$ for all $\omega\in\Omega$, $i=1,\ldots,k$ such that for all $Y \in L^0(\F_t;K)$ there exist $ \lambda_i\in L^0(\F_t;[-2^{i-1},2^{i-1}])$ with 
\beao
Y=\sum_{i=1}^k\lambda_iY^i\quad\text{a.s.}
\eeao
\end{lemma}
\begin{proof}
Let us prove this result by induction on the maximal value the random variable $\dim(K)$ can take. The base case $k=0$ is trivial. If $\dim(K) =0$ a.s. then $K=\{0\}$ a.s., and the statement holds.

Now assume that the claim has already been proven for all compact-valued $\F_t$-measurable random sets~$\wt{K}$ with $0\in \wt K$ and $\dim(\wt K)\le k-1$ a.s. for some $k\in \{1,\ldots,d\}$.
Let $K$ be a compact-valued $\F_t$-measurable random set with $0\in K$ and $\dim(K)\le k$ a.s. Let us first identify a random variable that maximizes the Euclidean norm in $K$. Since $K$ is compact-valued this is equivalent to maximizing the function 
\beao
h_K(x,\omega):=|x|-\delta_{K(\omega)}(x)=\left\{\begin{array}{ll} |x| &\text{if } x\in K(\omega), \\ -\infty &\text{otherwise}.\end{array}\right.
\eeao
Note that $-h_K$ is a normal integrand as defined in, e.g., \cite[Section~1.1.2]{pennanen.perkkioe.2024}. For every $\omega \in \Omega$ define $K_{max}(\omega):=\text{argmax}_xh_K(x,\omega)$. By \cite[Corollary 1.23]{pennanen.perkkioe.2024}, $K_{\max}$ is a closed-valued random set and by construction, we have
$K_{\max}(\omega)\subseteq K(\omega)$ for all $\omega \in \Omega$. Compact-valued random sets admit 
a measurable selection (\cite[Corollary 1.28]{pennanen.perkkioe.2024}) in the sense that there exists an $\F_t$-measurable random variable~$Y^k$ with $Y^k(\omega)\in K_{\max}(\omega)$ for all $\omega\in\Omega$.
Given such a $Y^k$, we define the functions 
\beao
\lambda(x,\omega):=\left\{\begin{array}{ll} \frac{\langle x,Y^k(\omega)\rangle}{|Y^k(\omega)|^2} &\text{if } Y^k(\omega)\neq 0, \\ 0 &\text{if }  Y^k(\omega)=0\end{array}\right.
\eeao
and $G(x,\omega):=x-\lambda(x,\omega)Y^k(\omega)$, which is  the projection of $x$ onto the orthogonal complement of $Y^k(\omega)$. It is easy to check that $G(\cdot,\omega)$ is continuous for every $\omega\in \Omega$, and that $G(x,\cdot)$ is measurable for every $x \in \bbr^d$.
By \cite[Theorem 1.9]{pennanen.perkkioe.2024}, the mapping $\omega \mapsto\wt{K}(\omega):=G(K(\omega),\omega)$ is measurable and $\dim(\wt{K})=(\dim(K)-1)\vee 0\le k-1$ a.s. By compactness of $K(\omega)$ for every $\omega \in \Omega$ and continuity of $G(\cdot,\omega)$ for every $\omega \in \Omega$, we obtain that $\wt{K}(\omega)=G(K(\omega),\omega)$ is compact for every $\omega\in \Omega$. Since $0\in \wt{K}$, we can apply the induction hypothesis to $\wt{K}$ and find $\F_t$-measurable random variables~$(\wt{Y}^i)_{i=1,\ldots,k-1}$ with $\wt{Y}^i(\omega)\in \wt{K}(\omega)$ for all $\omega\in\Omega$, $i=1,\ldots,k-1$ such that for every $\wt{Y}\in L^0(\F_t;\wt{K}),$ there exist $ \wt{\lambda}_i\in L^0(\F_t;[-2^{i-1},2^{i-1}])$, $i=1,\ldots,k-1$, with 
\beam\label{28.1.2026.1}
\wt{Y}=\sum_{i=1}^{k-1}\wt{\lambda}_i\wt{Y}^i \quad\mbox{a.s.}
\eeam
    Again using \cite[Theorem 1.9]{pennanen.perkkioe.2024} and the fact that the intersection of two closed-valued random sets is a closed-valued random set, we have that for each $i=1,\ldots,k-1$, $C^i(\omega):=(G^{-1}(\cdot,\omega)(\wt{Y}^i(\omega)))\cap K(\omega)$ defines a compact-valued random set. Since $\wt Y^i(\omega)\in \wt K(\omega)$ we have $C^i(\omega)\neq \emptyset$ for all $\omega \in \Omega$, $i=1,\ldots,k-1$.
    Again by \cite[Corollary 1.28]{pennanen.perkkioe.2024}, we can choose $\F_t$-measurable random variables~$Y^i$ with $Y^i(\omega)\in C^i(\omega)$ for all $\omega\in\Omega$ and $i=1,\ldots,k-1$. Now, let $Y$ be a representative of an element of $L^0(\F_t; K)$. The mapping $\omega\mapsto \lambda(Y(\omega),\omega)$ is an $\F_t$-measurable random variable that we denote by $\lambda(Y)$.
By construction, $Y-\lambda(Y)Y^k\in\wt{K}$ a.s., and by (\ref{28.1.2026.1}), there exist $\lambda_i\in L^0(\F_t;[-2^{i-1},2^{i-1}])$ for $i=1,\ldots,k-1$ such that
\beao
Y-\lambda(Y)Y^{k}=\sum_{i=1}^{k-1}\lambda_i\wt{Y}^{i}= \sum_{i=1}^{k-1}\lambda_i({Y}^{i}-\lambda({Y}^{i}){Y}^{k})=\sum_{i=1}^{k-1}\lambda_i{Y}^{i}-\left(\sum_{i=1}^{k-1}\lambda_i\lambda({Y}^{i})\right)Y^k\quad\mbox{a.s.}
\eeao
The $k$-th weight is set to $\lambda_k:=\lambda(Y)-\sum_{i=1}^{k-1}\lambda_i\lambda({Y}^{i})$. 
Since $Y^{k}$ maximizes the Euclidean norm in $K$, we have $|\lambda(Y)|\le 1$ a.s. and $|\lambda(Y^i)|\le 1$ for $i=1,\ldots,k$. This guarantees that 
$|\lambda_k|\le 1+ \sum_{i=1}^{k-1} 2^{i-1}=2^{k-1}$ a.s., which finishes the proof.
\end{proof}

\begin{figure}[htbp] \centering 
\begin{subfigure}{0.47\textwidth}
\centering
\begin{tikzpicture}[scale=1] 

 \draw[->] (0,0) -- (4,0) node[right] {$\beta^1$};
  \draw[->] (0,0) -- (0,3) node[above] {$\beta^2$};

  \draw[thick, blue, domain=-98:108, samples=100]
    plot ({2*cos(\x)*cos(30) - sin(\x)*sin(30)+1},
          {2*cos(\x)*sin(30) + sin(\x)*cos(30)+1});

  \draw[->, red, thick] (0,0) -- ({3.4*cos(37)},{3.4*sin(37)}) node[midway, below] {$Y^2$};

  \draw[-, red, dashed] ({-2*cos(127)},{-2*sin(127)}) -- ({3.4*cos(127)},{3.4*sin(127)});

\def\c{0.83}
   \draw[-, red, dashed] ({-\c*cos(127)},{-\c*sin(127)}) -- ({4*cos(37)-\c*cos(127)},{4*sin(37)-\c*sin(127)});
\def\d{1.22}
\draw[-, red, dashed] ({\d*cos(127)},{\d*sin(127)}) -- ({2*cos(37)+\d*cos(127)},{2*sin(37)+\d*sin(127)});

\draw[|-|, blue, thick] ({-\c*cos(127)},{-\c*sin(127)}) -- ({\d*cos(127)},{\d*sin(127)});

\draw[->, red, thick] (0,0) -- ({\d*cos(127)},{\d*sin(127)}) node[midway,left] {$\wt{Y}^1$};

\draw[->, red, thick] (0,0) -- (0.15,1.645) node[midway,right] {$Y^1$};

\def\e{(0.15,1.645)}

\end{tikzpicture}
\caption{The vector~$Y^2$ maximizes the Euclidean norm on $A_t$ and $\wt{Y}^1$ that on the orthogonal projection of $A_t$  along $Y^2$. $Y^1\in A_t$ is a vector that is projected on $\wt{Y}^1$.}\label{fig:left}
\end{subfigure}
\hfill 
\begin{subfigure}{0.47\textwidth}
\centering
\begin{tikzpicture}[scale=0.6] 

 \draw[->] (0,0) -- (4,0) node[right] {$\beta^1$};
  \draw[->] (0,0) -- (0,3) node[above] {$\beta^2$};

    \draw[thick, blue, domain=-98:108, samples=100]
    plot ({2*cos(\x)*cos(30) - sin(\x)*sin(30)+1},
          {2*cos(\x)*sin(30) + sin(\x)*cos(30)+1});

\draw[-, red, thick] ({3.4*cos(37)+0.15},{3.4*sin(37)+1.645}) -- ({-3.4*cos(37)+0.15},{-3.4*sin(37)+1.645}) node[above,left] {$Y^{1}-Y^{2}$};

\draw[-, red, thick] ({-3.4*cos(37)-0.15},{-3.4*sin(37)-1.645}) -- ({3.4*cos(37)-0.15},{3.4*sin(37)-1.645}) node[below,right] {$-Y^{1}+Y^{2}$};

\draw[-, red, thick] ({3.4*cos(37)-0.15},{3.4*sin(37)-1.645}) -- ({3.4*cos(37)+0.15},{3.4*sin(37)+1.645}) node[above] {$Y^{1}+Y^{2}$};

\draw[-, red, thick] ({-3.4*cos(37)+0.15},{-3.4*sin(37)+1.645}) --   ({-3.4*cos(37)-0.15},{-3.4*sin(37)-1.645}) node[below] {$-Y^{1}-Y^{2}$};

\end{tikzpicture}
\caption{The parallelogram with vertices $\pm Y^1 \pm Y^2$ does not contain the set $A_t$, but the larger one with vertices $\pm Y^1 \pm 2 Y^2$ does.}\label{fig:right}
\end{subfigure}
\caption{The set $A_t\subseteq \bbr_+^2$ for two non-redundant stocks. Its boundary is given by the blue curve.} \label{fig:gesamt}
\end{figure}
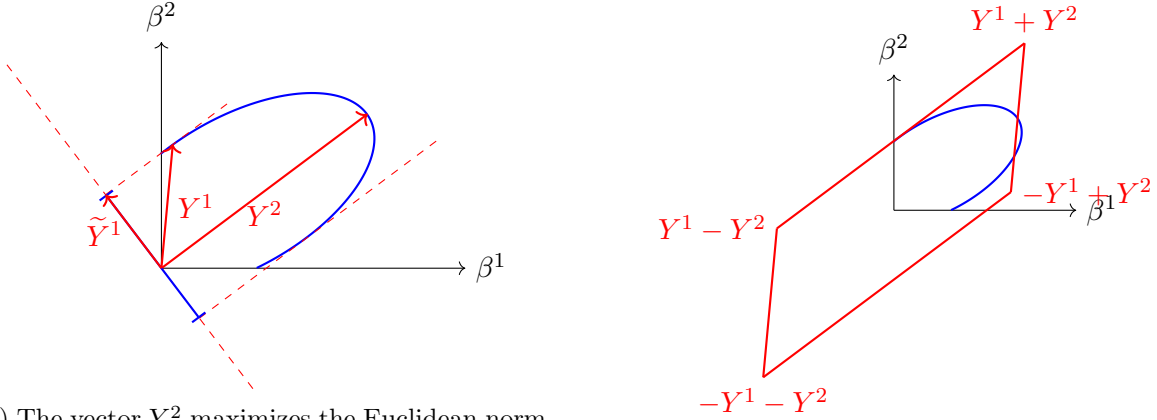

We are now ready to prove Lemma~\ref{lemma:KeyUtility}. 
\begin{proof}[Proof of Lemma~\ref{lemma:KeyUtility}] 
Let $t\in\{0,\ldots,T-1\}$. By Lemma~\ref{lemma:CompactSetA}, we can apply Lemma~\ref{lemma:BaseForCompSet} to the random set that is generated by $A_t$ and obtain: there are $\beta^i_t\in A_t$, $i=1,\ldots,d$, such that for all $\beta_t\in A_t$ there exist $\lambda_i\in L^0(\F_{t};[-2^{i-1},2^{i-1}])$ with $\beta_t=\sum_{i=1}^d\lambda_i\beta^i_t$.
For each $t=0,\ldots,T-1$ we define the one-period normalized strategy 
\beao
\wh \beta_t:=\sum_{i=1}^d\frac{2^{i-1}}{2^{d+1}-1}\beta^i_t\in L^0(\F_{t};\bbr^d_+).
\eeao
As a dynamic strategy, we consider $\wh N\in \mathcal{N}$ such that $\beta_t(\wh N)$ from (\ref{eq:normalizedStrat}), with initial capital $\wh x:=(2^{d+1}-1)^{T}x$, coincides with $q_t(\wh \beta_t)$ for each $t=0,\ldots,T-1$. By (\ref{eq:WealthProduct}), the wealth process reads
$$V^0_t(\wh{x},\wh{N})=\wh{x}\prod_{s=1}^t(1+r_s+\langle\wh\beta_{s-1},S_s-(1+r_s)S_{s-1}\rangle),\quad t=0,\ldots,T.$$
This means that the dominating wealth is constructed by starting with initial capital~$\wh{x}$ and investing the product of $\wh{\beta}_t$ and the wealth at the beginning of each period in the stocks.

Let $N\in \mathcal{A}^0(x)$ and define $\beta_t(N)$ as in (\ref{eq:normalizedStrat}) with initial capital~$x$. Under NA, one has $\beta_t(N)\in A_t$  for all $t=0,\ldots,T-1$. This implies that for a fixed $t=0,\ldots,T-1$ there exist $\lambda^N_i\in L^0(\F_t;[-2^{i-1},2^{i-1}])$ such that 
\beao
\beta_t(N)=\sum_{i=1}^d\lambda^N_i\beta^i_t.
\eeao
Next observe that we have the estimates:
\beam\label{eq:boundsforbetagains}
    |\langle\beta^i_t,S_{t+1}-(1+r_{t+1})S_{t}\rangle|\le 2(1+r_{t+1})+\langle\beta^i_t,S_{t+1}-(1+r_{t+1})S_{t}\rangle,\quad i=1,\ldots,d,
\eeam
which can be checked separately on the sets 
$\{0\ge \langle\beta^i_t,S_{t+1}-(1+r_{t+1})S_{t}\rangle\}$ and
$\{0<\langle\beta^i_t,S_{t+1}-(1+r_{t+1})S_{t}\rangle\}$, using the property $1+r_{t+1}+\langle\beta^i_t,S_{t+1}-(1+r_{t+1})S_{t}\rangle\ge 0$. We obtain
\beao
0&\le &1+r_{t+1}+\langle\beta_t(N),S_{t+1}-(1+r_{t+1})S_t\rangle\\
&= &1+r_{t+1}+\sum_{i=1}^d\lambda^N_i\langle\beta^i_t,S_{t+1}-(1+r_{t+1})S_t\rangle\\
&\le &1+r_{t+1}+\sum_{i=1}^d|\lambda^N_i||\langle\beta^i_t,S_{t+1}-(1+r_{t+1})S_t\rangle|\\
& \le &1+r_{t+1}+\sum_{i=1}^d2^{i-1}(2(1+r_{t+1})+\langle\beta^i_t,S_{t+1}-(1+r_{t+1})S_t\rangle)\\
& = & (1+r_{t+1})(2^{d+1}-1)+\langle\sum_{i=1}^d2^{i-1}\beta^i_t,S_{t+1}-(1+r_{t+1})S_t\rangle\\
& = & (2^{d+1}-1)(1+r_{t+1}+\langle\wh\beta_t,S_{t+1}-(1+r_{t+1})S_t\rangle),\quad t=0,\ldots,T-1,
\eeao
where the third inequality is by (\ref{eq:boundsforbetagains}). With the product representation of the wealth process from (\ref{eq:WealthProduct}), we can multiply both sides over all $t$ and obtain
\beam\label{eq:Majorant}
V_T^0(x,N)&=&x\prod_{t=0}^{T-1}(1+r_{t+1}+\langle\beta_t(N),S_{t+1}-(1+r_{t+1})S_t\rangle)\nonumber \\
&\le& x\prod_{t=0}^{T-1}(2^{d+1}-1)(1+r_{t+1}+\langle\wh\beta_t,S_{t+1}-(1+r_{t+1})S_t\rangle)\\
& =& V^0_T(\wh x,\wh N). \nonumber
\eeam
Applying (\ref{eq:Majorant}) to $N=0$ yields that $V^0_T(\wh x, \wh N)\ge x\prod_{t=1}^T(1+r_t)\ge x$ and thus $\E[U^-(V^0(\wh x,\wh N))]<~\infty$. Then, Proposition~\ref{lemma:IntegrabilityAllx} yields that $\E[U^+(V^0(\wh x,\wh N))]<\infty$. Since $\wh{x}$ and $\wh{N}$ do not depend on $N$, we are done.
\end{proof}

\begin{remark} 
Let us interpret our line of argument and compare it with that of R\'asonyi and Stettner~\cite{rasonyi.stettner.2006}.
We show that the (random) polytope generated by the points~$\sum_{i=1}^k \varepsilon_i 2^{i-1}Y^i$, $\varepsilon_i\in\{-1,1\}$, contains $A_t$ (cf. Lemma~\ref{lemma:BaseForCompSet} and Figure~\ref{fig:right}). Then, we can use that for nonnegative weights~$\lambda_i\in [0,2^{i-1}]$, one 
has $\lambda_i\langle Y^i,S_{t+1}-(1+r_{t+1})S_t\rangle\le 2^{i-1}\left[1+r_{t+1} + \langle Y^i,S_{t+1}-(1+r_{t+1})S_t\rangle\right]$ and for negative weights~$\lambda_i\in [-2^{i-1},0)$, one still has the weaker estimate $\lambda_i\langle Y^i,S_{t+1}-(1+r_{t+1})S_t\rangle\le 2^{i-1}\left[2(1+r_{t+1}) + \langle Y^i,S_{t+1}-(1+r_{t+1})S_t\rangle\right]$. Both hold since $\langle Y^i,S_{t+1}-(1+r_{t+1})S_t\rangle \ge -(1+r_{t+1})$ by $Y^i\in A_t$.
Putting these estimates together, we can construct a one-period strategy that requires a multiple of the initial capital, but then dominates (a.s.) all other strategies. The multiple only depends on $d$, which is the reason why the estimate can be extended in a straightforward way to multiperiod models.

In the proof of R\'asonyi and Stettner~\cite[Lemma~2.3]{rasonyi.stettner.2006}, a similar estimate is provided. But they estimate against a maximum of finitely many trading gains of strategies from $A_t$, which is generally not itself a trading gain.   
This is the reason why their approach to obtain an integrable majorant does not directly extend to the multiperiod model, and they apply other arguments instead. 

There is yet another estimate in the monograph of F\"ollmer and Schied~\cite[proof of Theorem~3.3]{foellmer.schied.2016}. Under the constraint that stock prices are nonnegative, a strategy with larger initial capital is constructed whose wealth dominates (a.s.) that of all strategies with initial capital~$1$. But, the required initial capital depends on the initial stock prices, which would be random from the second period onward, and thus, the result also does not extend directly to multiple periods.  
\end{remark}

If the market is frictionless and contains no redundant securities, and the utility function is strictly concave, then the maximizer is unique (\cite[Theorem~1.2]{rasonyi.stettner.2006}). In the following, we show that with taxes there can be multiple maximizers even if the corresponding frictionless market model has no redundant securities in the sense that  $\mathcal{R}_t=\{0\}$ for all $t=0,1,\ldots,T-1$. The minimalist example is a deterministic two-period model with $d=1$, $\alpha\in(0,1)$, $x=1$, $r>0$, $S_0=1$, $S_1=1+a$, and $S_2=(1+a)^2$, where $a\in ((1-\alpha)r,r)$ is given by the solution of the equation 
\beao
(1+a)^2(1-\alpha) + \alpha = (1+(1-\alpha)r)^2.
\eeao
Both investing the entire initial capital in the bank account and investing it in the stock are optimal, regardless of the utility function. However, the drawback of the example is that any debt-financed investment in stocks leads to sure losses. 
One might be led to conjecture that in a model with a single stock, maximizers are unique in those periods where there is a positive probability of outperforming the bank account.
We therefore provide a more sophisticated example showing that this is not the case. By concavity of $V^\alpha(x,\cdot)$ (Lemma~\ref{lemma:concave}(iv)), the optimal terminal wealth has to be unique if the utility function is strictly concave.
But, two different strategies (and their convex combinations) generate the same optimal terminal wealth.
\begin{example}[Non-uniqueness of maximizers]\label{30.1.2026.1}
Let $T=3$, $d=1$, $\alpha\in(0,1/9)$, $x=1$, $\Omega=\{\omega_1,\omega_2\}$ with $P(\{\omega_1\})=P(\{\omega_2\})=1/2$, and suppose that $\omega$ is revealed at time~$2$. The interest rate is constant~$r\in(0,1/3)$, and the price process of the single stock is given by $S_0=1$, $S_1=1+a$, $S_2(\omega_1)=(1+a)(1+4r)$, $S_2(\omega_2)=(1+a)(1-r)$, $S_3(\omega_1)=(1+a)(1+4r)(1+b)$, and $S_3(\omega_2)=0$. The parameters $a$ and $b$ are not yet specified. They should be strictly smaller than~$r$, but $r-a$ and $r-b$ should be small. 

This means that in the first and third period, the stock's deterministic return is lower than the interest rate, but on $\{\omega_1\}$, the advantage from deferring taxes would outweigh the difference.

First, we let $\lambda \ge 1$ and consider the strategy~$N^\lambda$ with $N^\lambda_{0,0}:=N^\lambda_{0,1}:=1$, $N^\lambda_{0,2}:=1_{\{\omega_1\}}$, $N^\lambda_{1,1}:=\lambda-1$, $N^\lambda_{1,2}:=3/4 (\lambda-1) 1_{\{\omega_1\}}$, and $N^\lambda_{2,2}:=0$. This means that the initial capital~$1$ is invested in the stock, $\lambda-1$~stocks are additionally purchased on credit at time~$1$, and at time~$2$ on $\{\omega_1\}$, the investor sells as many stocks as possible without triggering immediate tax payments. The stocks bought at $1$ that have lower book profits (if $a>0$) she sells first. If $r-b>0$ 
is small enough such that
\beam\label{15.12.2025.1.ck}
(1+4r)(1+b)(1-\alpha) + \alpha > (1+(1-\alpha)4r)(1+(1-\alpha)r), 
\eeam
this dominates a strategy that pays taxes already at time~$2$. We denote by $\un{r}$ the unique number $b$ that satisfies (\ref{15.12.2025.1.ck}) with equality. We note that $\un{r}\in((1-\alpha)r,r)$.

We want to construct a second strategy~$\wt{N}^\lambda$ with $\wt{N}^\lambda_{0,0}>1$ but $\wt{N}^\lambda_{1,1}<\lambda-1$ (for $\lambda>1$) that 
allows the investor to sell more stocks at time~$2$ without triggering immediate tax payments because of larger realized losses in the bank account. We want to choose the parameters $a$ and $b$ (that do not depend on $\lambda$) such that $V^\alpha(1,\wt{N}^\lambda)=V^\alpha(1,N^\lambda)$ for all $\lambda\ge 1$. The requirement that $V^\alpha(1,\wt{N}^\lambda)(\omega_2)=V^\alpha(1,N^\lambda)(\omega_2)$ implies that on $\{\omega_2\}$, the bank account positions of the strategies have to coincide after all shares are liquidated at time~$2$. This implication holds since both strategies pay no taxes up to time~$2$ (under the hypothesis that $\wt{N}^\lambda_{0,1}\ge 1$), and in the last period, they are only invested in the bank account. 
Setting the difference of the bank account positions at time~$2$ equal to zero yields the equation
\beam\label{15.12.2025.2.ck}
& & (\wt{N}^\lambda_{0,0}-1)\left[(1+a)(1-r) - (1+r)^2\right]\nonumber\\
& & + (\wt{N}^\lambda_{1,1}-(\lambda-1))\left[(1+a)(1-r) - (1+a)(1+r)\right] = 0.
\eeam
Now, we compare the realized losses of the candidate strategies~$\wt{N}^\lambda$ and $N^\lambda$ up to time~$2$ on $\{\omega_1\}$. By (\ref{15.12.2025.2.ck}), the difference reads 
\beam\label{15.12.2025.3.ck}
& & (\wt{N}^\lambda_{0,0}-1)(2r+r^2) + (\wt{N}^\lambda_{1,1}-(\lambda-1))(1+a)r\nonumber\\
& & = (\wt{N}^\lambda_{0,0}-1)\left[2r+r^2 - (1+a)r\frac{(1+a)(1-r) - (1+r)^2}{(1+a)(1-r) - (1+a)(1+r)}\right]\nonumber\\
& & =: (\wt{N}^\lambda_{0,0}-1)\left[2r+r^2 - (1+a)r g(a)\right]\nonumber\\
& & =: (\wt{N}^\lambda_{0,0}-1)L(a).
\eeam
We observe that $L(a)$ is positive since $2r+r^2 - (1+a)r g(a)
=2r+r^2 - (3r-a+r^2+ar)/2>0$ for $a\in((1-\alpha)r,r)$.
By (\ref{15.12.2025.3.ck}) and again (\ref{15.12.2025.2.ck}), we can express the difference~$V^\alpha(1,\wt{N}^\lambda)(\omega_1) - V^\alpha(1,N^\lambda)(\omega_1)$ as a function~$f$ of $\wt{N}^\lambda_{0,0}-1>0$. Crucial is that the function does not depend on $\lambda$.
After being taxed at time $3$, there are gains from three investment intervals to compare, from $0$ to $3$, from $1$ to $2$, and from $1$ to $3$: 
\beao
& & f(\wt{N}^\lambda_{0,0}-1)\\ 
& & = (\wt{N}^\lambda_{0,0}-1)(1-\alpha)\left[(1+a)(1+4r)(1+b) - (1+r)^3\right]\\
& & + \frac{(\wt{N}^\lambda_{0,0}-1)L(a)}{(1+a)4r}(1-\alpha)(1+a)\left[(1+4r)(1+r)-(1+r)^2\right]\\ 
& & + \left[-(\wt{N}^\lambda_{0,0}-1)g(a)-\frac{(\wt{N}^\lambda_{0,0}-1)L(a)}{(1+a)4r}\right](1-\alpha)(1+a)\left[(1+4r)(1+b)-(1+r)^2\right]\\
& & =: (\wt{N}^\lambda_{0,0}-1) A + \frac{(\wt{N}^\lambda_{0,0}-1)L(a)}{(1+a)4r} B + \left[-(\wt{N}^\lambda_{0,0}-1)g(a)-\frac{(\wt{N}^\lambda_{0,0}-1)L(a)}{(1+a)4r}\right] C.
\eeao
Since $f$ is linear, it vanishes if its slope, which depends on $a$ and $b$, is zero. For $a=r$ and $b<r$, one has that $g(a)=1$, $A=C$, and $B>A$. 
Thus, the slope is positive. 
For $b=r$ and $a<r$, we have $B=C$ and $A<C$. Together with $g(a)\ge 1$ and $C>0$, this implies that the slope is negative.
Now, we vary $a$ between $\un{r}$ and $r$ and set $b=\un{r}+r-a$.  As the slope is negative for $a=\un{r}$ and positive for $a=r$, there exists an intermediate value such that $f=0$. We take such an intermediate value for $a$ and $b$. Next, the strategy $\wt{N}^\lambda$ is formally defined.
First, we define $\wt{N}^\lambda_{0,0}$ by
\beam\label{13.2.2026.1}
(\wt{N}^\lambda_{0,0}-1)g(a) + \frac{(\wt{N}^\lambda_{0,0}-1)L(a)}{(1+a)4r} = 3/4(\lambda-1).
\eeam
Given $\wt{N}^\lambda_{0,0}$, $\wt{N}^\lambda_{1,1}$ is derived from (\ref{15.12.2025.2.ck}). 
We complete the definition by setting $\wt{N}^\lambda_{0,1}:=\wt{N}^\lambda_{0,0}$, $\wt{N}^\lambda_{0,2}:=\wt{N}^\lambda_{0,0}1_{\{\omega_1\}}$, $\wt{N}^\lambda_{1,2}:=0$, and 
$\wt{N}^\lambda_{2,2}:=0$. By (\ref{13.2.2026.1}), the formal definition of $\wt{N}^\lambda$ is in line with its earlier description. 
If $\omega_1$ occurs, just the stocks with the lower book profits are completely liquidated at time~$2$.
We arrive at $V^\alpha(1,\wt{N}^\lambda)=V^\alpha(1,N^\lambda)$ (and by linearity of $f$, all convex combinations of $N^\lambda$ and $\wt{N}^\lambda$ have the same terminal wealth as well). We observe that $N^\lambda$ and $\wt{N}^\lambda$ differ for $\lambda>1$.

Now, we consider a strategy~$N$ with $\lambda:=N_{0,0}<1$. Either $N$ can be strictly dominated by another strategy or we must have that 
$N_{1,1}=N_{2,2}=0$, $N_{0,1}=N_{0,0}$, and $N_{0,2}=N_{0,0}1_{\{\omega_1\}}$. We denote this strategy by $N^\lambda$ for $\lambda<1$. Indeed, by the choice of $a,b$, an investor who wants to buy debt-financed stocks is indifferent between buying a part of them already at time~$0$ or buying the entire position only at time~$1$. This makes a strategy with $N_{0,0}<x=1$ and $N_{1,1}>0$ strictly suboptimal because 
the situation is similar but the strategy has to pay taxes on the gains in the bank account already at time~$1$. Since strategies with $N_{0,0}\ge 1$ are also dominated, we arrive at
\beao
\{V^\alpha(1,N)\ :\ N\in \mathcal{N}\}-L^0(\F_T; \bbr_+)  = \{V^\alpha(1,N^\lambda)\ :\ \lambda\in\bbr_+\}-L^0(\F_T; \bbr_+),
\eeao
and a utility maximization problem boils down to a one-dimensional maximization problem over $\lambda$. 
Considered as functions of $\lambda$, the terms $V^\alpha(1,N^\lambda)(\omega_1)$ and $V^\alpha(1,N^\lambda)(\omega_2)$ are differentiable at $\lambda=2$,
and the derivatives at this point, denoted by $y_1$ and $y_2$, satisfy $y_1>0>y_2$ and $y_1 > -y_2$. The latter follows from $(1-\alpha)3r(1+(1-\alpha)r) > 2r(1+r)$.
In addition, one has $V^\alpha(1,N^2)(\omega_1) > V^\alpha(1,N^2)(\omega_2)=(1+a)(1-3r)(1+r)>0$. 
Now, we can take an arbitrary utility function~$U$ that satisfies Assumption~\ref{14.2.2026.2} and that is differentiable at $V^\alpha(1,N^2)(\omega_1)$, $V^\alpha(1,N^2)(\omega_2)$ with 
\beam\label{14.2.2026.1}
\frac12 y_1 U'(V^\alpha(1,N^2)(\omega_1)) + \frac12 y_2 U'(V^\alpha(1,N^2)(\omega_2)) = 0. 
\eeam
By construction of the parameters $a,b$, we have that $V^\alpha(1,N^{(1-\beta)\lambda_1+\beta\lambda_2})\ge V^\alpha(1,(1-\beta)N^{\lambda_1}+\beta N^{\lambda_2})$ for all $\lambda_1,\lambda_2\ge 0$, $\beta\in[0,1]$. Together with Lemma~\ref{lemma:concave}(iv), this implies that the 
mappings~$\lambda\mapsto V^\alpha(1,N^\lambda)(\omega_i)$, $i=1,2$, are concave. Therefore, (\ref{14.2.2026.1}) implies that $\E[U(V^\alpha(1,N^2))]=\sup_{N\in\mathcal{N}}\E[U(V^\alpha(1,N))]$. Since $N^2$ and $\wt{N}^2$ generate the same terminal wealth, we are done.
\end{example}

\section{Closedness}\label{23.8.2026.1.ck}

\subsection{Sufficient conditions for closedness}\label{22.8.2026.1.ck} 

On a finite probability space, the set~${\mathcal{V}^\alpha(x)}$ from (\ref{17.9.2026.1}) is always closed---even if NA does not hold.
\begin{proposition}\label{28.11.2025.1.ck}
Let $|\Omega|<\infty$ and $x\in\bbr$. Then, ${\mathcal{V}^\alpha(x)}$ is closed (with respect to the supremum norm).
\end{proposition}
\begin{proof}
We consider a family of artificial linear tax rules with tax rebates for realized losses. For each tax rule, there is a random subset of time points in $\{1,\ldots,T\}$, and outside this random set no taxes are paid. To formalize this, let $(\tau_t)_{t=1,\ldots,T}$ be an adapted $\{0,\ldots,T\}$-valued process with $\tau_1\le \ldots\le \tau_T$, $\tau_t\le t$, and $\tau_{\tau_t}=\tau_t$ for $\tau_t\ge 1$. For a given strategy~$N$, we define the tax process by $\Pi^{(\tau)}_t:=\alpha G_{\tau_t}$, where $G$ is given in (\ref{13.11.2025.1.ck}), and the corresponding self-financing position in the bank account with initial capital~$x$ is given in (\ref{13.11.2025.3.ck}). The latter is denoted by $\eta^{(\tau)}(N)$. We note that---as in the original model---the recursive construction is explicit since $G_t$ does not depend on $\eta_t$. The random variable~$\tau_t$ is interpreted as the last time in $\{0,1,\ldots,t\}$ at which taxes on accumulated realized gains and losses are paid (since $G_0=0$, $\tau_t=0$ means that no taxes are paid up to time~$t$). Of course, it can happen that gains and losses are not taxed at all under such a rule. Let us show that
\beam\label{4.10.2025.1.ck}
V^\alpha(x,N) = \min_{\tau}\eta^{(\tau)}_T(N),
\eeam
where the minimum is taken over the finite set of all tax rules described above. The inequality ``$\ge$'' is obvious since for a fixed strategy~$N$, the LUL tax process coincides with the tax process of the artificial tax rule with $\tau^\star_t:=\min\{s\in\{0,\ldots,t\} : G_s(N)=\max_{u=0,\ldots,t}G_u(N)\}$. 
To show ``$\le$'', we start with an arbitrary tax rule~$\tau$ and define the tax rule
\beao
\tau^1_t:=\left\{
\begin{array}{ll} 
\tau_t & \textrm{if}\ \tau_t\ge 2\\
1 & \textrm{if}\ \tau_t\le 1\ \mbox{and}\ G_1>0\\
0 & \textrm{if}\ \tau_t\le 1\ \mbox{and}\ G_1\le 0\\
\end{array},\quad t\ge 1.\right.
\eeao
By $\{\tau^1_1=1\}=\{G_1>0\}$ and by $r\ge 0$, it follows that $\eta^{(\tau^1)}_T(N) \le \eta^{(\tau)}_T(N)$. Analogously, we construct the tax rule
\beao
\tau^2_t:=\left\{
\begin{array}{ll} 
\tau^1_t & \textrm{if}\ \tau^1_t\ge 3\ \mbox{or}\ t=1\\
2 & \textrm{if}\ \tau^1_t\le 2\ \mbox{and}\ t \ge 2\ \mbox{and}\ G_2>0\vee G_1\\
\tau^1_1 & \textrm{if}\ \tau^1_t\le 2\ \mbox{and}\ t \ge 2\ \mbox{and}\ G_2\le 0\vee G_1\\
\end{array},\quad t\ge 1,\right.
\eeao
where the process~$G$ is that associated with the tax rule~$\tau^1$.
After $T$ steps we arrive at $\tau^\star$, and $V^\alpha(x,N) = \eta^{(\tau^\star)}_T(N)\le \eta^{(\tau)}_T(N)$ is proven for 
all $\tau$.

Now, we identify a strategy~$N\in\mathcal{N}$ with the vector consisting of the nonnegative values of $N_{t,t,j}$, $t=0,\ldots,T-1$, $j=1,\ldots,d$, and $N_{i,t-1,j} - N_{i,t,j}$, $i=0,\ldots,T-1$, $t=i+1,\ldots,T$,  $j=1,\ldots,d$, on every atom of the $\sigma$-algebra~$\mathcal{F}_t$ (see, e.g., Kallsen and Muhle-Karbe~\cite[proof of Theorem~3.2]{kallsen.muhle-karbe.2011} for such an identification). The inequality $N_{i,t-1,j} - N_{i,t,j}\le N_{i,i,j} -\sum_{u=i+1}^{t-1}(N_{i,u-1,j}-N_{i,u,j})$ for $t\ge i+1$ has to hold for any decreasing sequence of atoms of $\mathcal{F}_i,\ldots,\mathcal{F}_t$ (with equality for $t=T$). With this identification, $\mathcal{N}$ is a polyhedral convex cone in $\bbr^k$ with some suitable $k\in\bbn$, and    
the mapping $N\mapsto \eta^{(\tau)}_T(N) - \eta^{(\tau)}_T(0)$ from $\mathcal{N}\subseteq\bbr^k$ to $\bbr^{|\Omega|}$ is linear for $\tau$ fixed. The latter follows from a decomposition of $\eta^{(\tau)}_T(N)$ analogously to K\"uhn~\cite[Equations~(3.2)/(3.3)]{kuehn.2019} that is for the linear tax rule with $\tau_t=t$. We leave it to the reader to write this down in detail.
Since there are only finitely many tax rules, the 
set~$\{(N,V)\in\mathcal N\times L^0(\F_T;\mathbb R)\ :\ V\le \eta_T^{(\tau)}(N)\ \text{for all }\tau
\}$ is a polyhedron. By (\ref{4.10.2025.1.ck}), its projection onto the second component coincides with $\mathcal{V}^\alpha(x)$. The projection of a polyhedron is again a polyhedron (see, e.g., Rockafellar~\cite[Theorem~19.3]{rockafellar.1970}), and hence $\mathcal{V}^\alpha(x)$ is closed.
\end{proof}
There is another sufficient condition for closedness that is easy to verify:
\begin{proposition}\label{31.8.2026.1} If the market satisfies the NA condition and if there are no redundant stocks in the sense that $\mathcal{R}_t=\{0\}$ for all $t=0,\dots,T-1$, then $\mathcal{V}^\alpha(x)$ is closed with respect to the convergence in probability.
\end{proposition}    
Finally, we weaken the no-redundancy condition by assuming that any terminal wealth a taxable investor can achieve she can 
achieve without debt-financed reversible investments. To make this precise, we define a smaller class of strategies by
\beao
\mathcal{N}'(x) & := & \{N\in \mathcal{N}\ :\  \forall t=0,1,\ldots,T-1\  \exists \mbox{ maximal } \gamma_t\in \{\gamma\in \mathcal{R}_t\ :\ \gamma\le N_{t,t}\ \mbox{a.s.}\} \mbox{ s.t. }\\
&& \langle\gamma_t,S_t\rangle\le ((1+r_t 1_{(t\ge 1)})\eta_{t-1} 
+ \langle\sum_{i=0}^{t-1}(N_{i,t-1}-N_{i,t}), S_t\rangle - 1_{(t\ge 1)}(\Pi_t-\Pi_{t-1}))^+\}.
\eeao    
The term in $(\ldots)^+$ is the amount in the bank account at time~$t$ after sell orders are executed and taxes are paid but before new shares are purchased. 
It is sufficient to have the estimate for an arbitrary reversible strategy~$\gamma_t\le N_{t,t}$ that cannot be strictly dominated by another one.
The reversible part~$p_t(N_{t,t})$ from Lemma~\ref{30.12.2025.1} is only an example for such a maximal element~$\gamma_t$.
\begin{proposition}\label{31.8.2026.2}    
    If the market satisfies the NA condition, $S\ge 0$ and $\mathcal{V}^\alpha(x)=\{V^\alpha(x,N')\ :\ N'\in \mathcal{N}'(x)\}-L^0(\F_T;\bbr_+)$, 
    then $\mathcal{V}^\alpha(x)$ is closed with respect to the convergence in probability.
\end{proposition}
Let us consider the following minimalist deterministic example.
\begin{example} 
Let $T=2$, $d=1$, $r\in \bbr_+\setminus\{0\}$, $S_t:=(1+r)^t$, and $x\in \bbr_+\setminus\{0\}$. We have $\mathcal{R}_t\not=\{0\}$, but the conditions of Proposition~\ref{31.8.2026.2} are satisfied since $N_{0,0}=N_{0,1}=x$ and $N_{1,1}=0$ is a dominant strategy. This means that the stock can be used to defer taxes at time~$1$, but a debt-financed purchase of the stock would not provide an additional advantage.
\end{example}

\begin{proof}[Proof of Proposition~\ref{31.8.2026.2}]
By NA and $S\ge 0$, zero is an absorbing state of each stock price process~$S^j$. This allows us to restrict ourselves throughout the proof to strategies with $N_{t,t,j}=0$ on $\{S^j_t=0\}$. Under this restriction, we take an $N\in \mathcal{N}'(x)$ and an associated process~$(\gamma_t)_{t=0,\dots,T-1}$
and estimate
\beam\label{eq:BoundReversible}
(1\wedge \inf\{S^j_t\ :\ j\ \mbox{s.t.}\ S^j_t>0\})\langle \gamma_t,1\rangle \le \langle\gamma_t,S_t\rangle
& \le & \left((1+r_t1_{(t\ge 1)})\eta_{t-1}+\langle\sum_{i=0}^{t-1}N_{i,t-1},S_t\rangle\right)^+\nonumber\\
& \le & (V^0_t(x,N))^+.
\eeam
Let $(N^n,\mu_n)_{n\in \bbn}\subseteq \mathcal{N}'(x)\times L^0(\F_T; \bbr_+)$ be a sequence such that $V^\alpha(x,N^n)-\mu_n \to f$ in probability. By passing to a subsequence we can and do assume a.s.-convergence. By definition of $\mathcal{N}'(x)$, for each $n\in \bbn$ there exist maximal elements~$(\gamma^n_t)_{t=0,\dots,T-1}$ with $\mathcal{R}_t\ni\gamma_t^n\le N^n_{t,t}$ such that (\ref{eq:BoundReversible}) holds.
Let us prove that $\sup_{n\in \bbn}|N^n_{s,s}|<\infty$ a.s. for all $s\le t$ by induction on $t=0,\dots,T-1$. The case $t=0$ is a special case of the induction step.
Suppose we already proved that $\sup_{n\in \bbn}|N_{s,s}^n|<\infty$ for all $s<t$. Since $V^0_t(x,N^n)$ only depends on $x$ and $N^n_{s,u}$ for $0\le s\le u\le t-1$, we get with the induction hypothesis $\sup_{n\in \bbn}|V^0_t(x,N^n)|<\infty$. From (\ref{eq:BoundReversible}) it immediately follows that $\sup_{n\in \bbn}|\gamma_t^n|<\infty$. Let us show that $\sup_{n\in \bbn}|N_{t,t}^n-\gamma_t^n|<\infty$. We define two sequences of strategies 
\beao
{N}^{(1),n}_{s,u}:=\left\{\begin{array}{ll} \sum_{i=0}^uN_{i,u}^n & \text{ if } 0\le s = u\neq t\\
\sum_{i=0}^uN^n_{i,u}-\gamma_t^n& \text{ if } 0\le s = u= t\\ 0 & \text{ otherwise}\end{array}\right., \quad
{N}^{(2),n}_{s,u}:=\left\{\begin{array}{ll} \sum_{i=0}^uN_{i,u}^n & \text{ if } 0\le s = u\neq t\\
N^n_{t,t}-\gamma_t^n& \text{ if } 0\le s = u= t\\0 & \text{ otherwise}.\end{array}\right.
\eeao
We know that $V^0(x,N^n)=V^0(x,N^{(1),n})$, thus $\liminf_{n\to\infty} V^0(x,N^{(1),n})>-\infty$. The strategies $N^{(1),n}$ and $N^{(2),n}$ differ only in their $(t,t)-$component, where $N^{(1),n}_{t,t}-N^{(2),n}_{t,t}=\sum_{i=0}^{t-1}N^n_{i,t}$. 
Together with the induction hypothesis, we obtain
\beao
\sup_{n\in \bbn}|V^0(x,N^{(1),n})-V^0(x,N^{(2),n})|&\le& \sup_{n\in \bbn}\prod_{s=t+2}^T(1+r_s)|\sum_{i=0}^{t-1}N^n_{i,t}||S_{t+1}-(1+r_{t+1})S_{t}|<\infty
\eeao
and therefore $\liminf_{n\to\infty} V^0(x, N^{(2),n})>-\infty$. By Lemma~\ref{lemma:L0-bound nonreversible}(iii), it follows that $\sup_{n\in \bbn}|q_t(N_{t,t}^n-\gamma_t^n)|<\infty$. By maximality of $\gamma_t^n$, we have that $N^n_{t,t}-\gamma_t^n\in \mathcal{P}_t$, i.e., $q_t(N_{t,t}^n-\gamma_t^n)=N_{t,t}^n-\gamma_t^n$ for all $n\in \bbn$. We conclude that 
\beam\label{31.8.2026.3}    
\sup_{n\in \bbn} \max_{i=0,\dots,T-1,j=1,\dots,d}N^n_{i,i,j}<\infty.
\eeam
 By Lemma~\ref{lemma:komlos}, applied successively to each component of $(N^{n})_{n\in\bbn}$, there exists a sequence of convex weights $(\lambda_n^n,\ldots,\lambda^n_{J_n})_{n\in \mathbb{N}}$ such that
\beam\label{eq:converStrat}
 \sum_{k=n}^{J_n}\lambda_k^n N^{k} \to {N} \in \mathcal{N}\quad\text{a.s.}\quad\mbox{as}\ n\to \infty.
 \eeam
The concavity of $V^\alpha(\cdot,\cdot)$ now implies that
\beao
V^\alpha(x,\sum_{k=n}^{J_n}\lambda_k^nN^{k}) \ge  \sum_{k=n}^{J_n}\lambda_k^nV^\alpha(x,N^{k})
 \ge  \sum_{k=n}^{J_n}\lambda_k^n(V^\alpha(x,N^{k})-{\mu}_k)\quad\quad \mbox{for all\ }n\in\bbn.
\eeao
The RHS converges to $f$ a.s. as $n\to\infty$ and by (\ref{eq:converStrat})
the LHS converges to $V^\alpha(x,N)$ a.s. as $n\to \infty$. We obtain that
$V^\alpha(x,N)\ge f$, hence $f\in \mathcal{V}^\alpha(x)$.
\end{proof}

\begin{proof}[Proof of Proposition~\ref{31.8.2026.1}]
Under the assumptions, (\ref{31.8.2026.3}) can be derived directly from Lemma~\ref{lemma:L0-bound nonreversible}(iii). From there, we can follow the rest of the proof of Proposition~\ref{31.8.2026.2}, and the assumption $S\ge 0$ is not needed.  
\end{proof}

\subsection{Characterization of the closure}

The following theorem sheds some light on the closedness property. We consider a slightly enlarged class of strategies, including certain 
{\em fictitious} strategies. In particular, we allow an anticipative selection of the original strategies whenever---and only when---dominance occurs among them. For given initial capital $x\in\bbr$, we define
\beam\label{12.8.2026.1.ck}
\wt{\mathcal{N}}(x) & := & \Big\{\sum_{k=1}^\infty 1_{A_k} N^k\ :\ (N^k)_{k\in\bbn}\subseteq\mathcal{N}\ \mbox{sequence with}\ V^\alpha(x,N^1)
\le V^\alpha(x,N^2)\le \ldots\nonumber\\
& & \qquad\qquad\qquad\quad\mbox{and}\ (A_k)_{k\in\bbn}\subseteq \F_T\ \mbox{partition of}\ \Omega\Big\} 
\eeam
and $V^\alpha(x,\wt{N})$ accordingly. Observe that the resulting fictitious strategies are not genuine limit strategies, since for fixed $\omega$ a finite~$k$ has to be chosen. 

\begin{theorem}\label{main:closedness}         
Let $x\in\bbr$. Under NA we have that
\beao
\overline{\mathcal{V}^\alpha(x)}=\{V^\alpha(x,\wt N)\ :\ \wt N\in \wt{\mathcal{N}}(x)\}-L^0(\F_T;\bbr_+),
\eeao
where $\overline{\mathcal{V}^\alpha(x)}$ denotes the closure of ${\mathcal{V}^\alpha(x)}$ with regard to convergence in probability.\\
Consequently, under NA, $\mathcal{V}^\alpha(x)$ is closed in probability if and only if for any $\wt{N}\in\wt{\mathcal{N}}(x)$ there is an $N\in\mathcal{N}$ such that $V^\alpha(x,N)\ge V^\alpha(x,\wt{N})$ a.s.
\end{theorem}
\begin{remark}
Enlargement~(\ref{12.8.2026.1.ck}) is in the spirit of Jacka, Berkaoui, and Warren~\cite{jacka.berkaoui.warren.2008}, who
consider an adjusted bid-ask spread in models with proportional transaction costs. 
Tax models are more complex because the same trade at different times can yield different terminal portfolios, even when discounted trading prices coincide. Enlargement~(\ref{12.8.2026.1.ck}) has the advantage that it can be formulated for a wide range of stochastic optimization models. It depends only on the information structure, but not on the fine structure of the strategies. Under this constraint, it is the smallest enlargement that one could reasonably expect to yield additional attainable terminal wealth.
\end{remark}
\begin{remark}
Theorem~\ref{main:closedness} shows that closedness with the original strategies
fails if and only if there exists an anticipative selection 
from a sequence of dominating strategies that cannot be dominated by another strategy from the original class.
For example, this immediately shows that on a finite probability space, one has closedness with the original strategies~(namely, a given strategy from $\wt{\mathcal{N}}(x)$ with a representation as in (\ref{12.8.2026.1.ck}) can be dominated by $N^{k_0}\in\mathcal{N}$ with $k_0:=\sup\{k\in\bbn\ :\ A_k\not=\emptyset\}<\infty$).
However, on a finite probability space, closedness holds even without NA (Proposition~\ref{28.11.2025.1.ck}). On the other hand, Theorem~\ref{main:closedness} can fail when NA is violated and the probability space is infinite. Since this is not obvious, we provide the following (counter)example. 
\end{remark}
\begin{example}[$\exists$\ arbitrage and $\{V^\alpha(x,\wt{N})\ :\ \wt{N}\in\wt{\mathcal{N}}(x)\}-L^0(\F_T;\bbr_+)$ not closed]  
Let $\Omega=\bbn$, $T=1$, $\F_0=\{\emptyset, \Omega\}$, $\F_1=2^\Omega$, $\alpha=0$, $x=0$, and $r=0$. All states have positive probabilities. 
There are two assets whose price increments $\Delta S_1^j := S_1^j-S_0^j$, $j=1,2$, satisfy $\Delta S^1_1(k)=-1$, $\Delta S^2_1(k)=1$ for $k$ odd and 
$\Delta S^1_1(k)=1$, $\Delta S^2_1(k)=-1+1/k$ for $k$ even. This means that the assets hedge each other, but in every even state~$k$, asset~$2$ is worth an additional $1/k$ compared to the perfect hedge. The arbitrage strategies (in the original class) are given by the positive multiples of $N_{0,0}=(1,1)$. 
We note that for every $\eps>0$, the strategy~$N_{0,0}=(1,1+\eps)$ is not an arbitrage, as it leads to losses in the even states~$k$ with~$k>(1+\eps)/\eps$. But, the probability of a loss tends to zero as $\eps\downarrow 0$.
In this tax-free example, strict dominance can only be obtained by adding arbitrage strategies.
Thus, the strategies from $\wt{\mathcal{N}}(0)$ have to be of the form~$\wt{N}_{0,0}(k)=(a_1+f(k),a_2+f(k))$ for some $a_1,a_2\in\bbr_+$ and a 
function~$f:\bbn\to\bbr_+$. 

Now, we consider the random variable~$H$ defined by $H(k):=k$ for $k$ odd and $H(k):=0$ for $k$ even. 
Since $N^n_{0,0}:=(n,n+\sqrt{n})=n(1,1+1/\sqrt{n})$ are strategies in the original class with $V^0(0,N^n)(k)=\sqrt{n}\to\infty$ as $n\to\infty$ for $k$ odd and 
$(V^0(0,N^n)(k))^- = (n/k + \sqrt{n}(-1+1/k))^- \to 0$ as $n\to\infty$ for $k$ even, we have $H\in\ov{\mathcal{V}^0(0)}$.
On the other hand, $H$ cannot be superhedged with zero initial capital, not even with an extended strategy~$\wt{N}_{0,0}=(a_1+f(k),a_2+f(k))_{k\in\bbn}$, since the terminal wealth for $k$ odd is $a_2-a_1$ and does not increase with $k$, thus $H\notin\{V^0(0,\wt{N})\ :\ \wt{N}\in\wt{\mathcal{N}}(0)\}-L^0(\F_T;\bbr_+)$.
\end{example}

The proof of Theorem~\ref{main:closedness} requires some preparation. First, we introduce \textit{tax-efficiency}, which can be formulated pathwise.
The pathwise construction is essential since it has to be applied to anticipative strategies.
A vector~$b=(b_t^j)_{t=0,\dots ,T-1,j=1,\dots ,d}\in \bbr_+^{dT}$ represents the total holdings. We define a Carathéodory function  $(b,\omega)\mapsto N^{\rm eff}(b,\omega)\in \bbr_+^{T(T+1)/2d}$, which is continuous in $b$ and $\F_T$-measurable in $\omega$, as follows: For $j=1,\dots ,d$, set ${N}^{\rm eff}_{0,0,j}(b,\omega):=b_0^j$ and given ${N}^{\rm eff}_{s,t-1,j}(b,\omega)$ for all $s \in \{0,\dots,t-1\}$ we define 
\beam\label{eq:recursion taxeff1}
{N}^{\rm eff}_{s,t,j}(b,\omega):=1_{\{S^j_t\ge S^j_s\}}(\omega)\left({N}^{\rm eff}_{s,t-1,j}(b,\omega)-\left((b_t^j-b_{t-1}^j)^--\sum_{i=s+1}^{t-1}{N}^{\rm eff}_{i,t-1,j}(b,\omega)\right)^+\right)^+
\eeam
for $s=0,\dots,t-1$ and
\beam\label{eq:recursion taxeff2}
{N}^{\rm eff}_{t,t,j}(b,\omega):=b_t^j-b_{t-1}^j+\sum_{s=0}^{t-1}({N}^{\rm eff}_{s,t-1,j}(b,\omega)-{N}^{\rm eff}_{s,t,j}(b,\omega)).
\eeam
This means that a stock is sold (and possibly rebought) when its price falls below the purchase price. If a position is reduced, the shares with the shortest holding periods are sold first.
\begin{lemma}\label{effiz}
For a given $b=(b_t^j)_{t=0,\dots T-1,j=1,\dots d}\in \bbr_+^{dT}$, $N^{\rm eff}(b,\omega)$ satisfies $\sum_{s=0}^tN^{\rm eff}_{s,t}(b,\omega)=b_t$ for all $t=0,\dots,T-1$, and for every other strategy $\wt n=(\wt n_{i,t,j})_{0\le i \le t\le T-1,j=1,\dots,d}\in \bbr_+^{T(T+1)/2d}$ 
nonincreasing in $t$ that satisfies $\sum_{s=0}^t\wt{n}_{s,t}=b_t$ for all $t=0,\dots,T-1$, one has $V^{\alpha}(x, N^{\rm eff}(b,\omega))(\omega)\ge V^{\alpha}(x,\wt n)(\omega)$ for every $\omega\in \Omega$.
\end{lemma}
\begin{proof}

We fix an $\omega\in \Omega$ and omit it in the notation. Let $b=(b_t^j)_{t=0,\dots ,T-1,j=1,\dots d}\in \bbr_+^{dT}$ and let $\wt n=(\wt n_{i,t,j})_{0\le i\le t\le T-1,j=1,\dots,d}\in \bbr_+^{T(T+1)/2d}$ be nonincreasing in $t$ such that $\sum_{s=0}^t\wt{n}_{s,t}=b_t$ for all $t=0,\dots,T-1$. We define $\wh n:=N^{\rm eff}(b,\omega)$. 
The first statement of the lemma is immediate. Let us prove that $V^{\alpha}(x, \wh n)\ge V^{\alpha}(x,\wt n)$. By K\"uhn and Ulbricht~\cite[Theorem A.1]{kuehn.ulbricht.2015}, we get
\beam\label{eq:stockgains}
    \sum_{s=1}^t\sum_{i=0}^{s-1}\langle \wh n_{i,s-1}-\wh n_{i,s},S_s-S_i\rangle\le \sum_{s=1}^t\sum_{i=0}^{s-1}\langle \wt n_{i,s-1}-\wt n_{i,s},S_s-S_i\rangle
\eeam
for $t=0,1,\dots,T$, that is, the accumulated realized stock gains are minimized by $N^{\rm eff}(b,\omega)$ among all strategies with the same total holdings---simultaneously for all $t$. Let us prove by induction on $t=0,\ldots,T$ that $\eta_t(\wt n)\le \eta_t(\wh n)$. The base case $t=0$ is trivial.
Suppose that we have shown $\eta_s(\wt n)\le \eta_s(\wh n)$ for every $ s<t$. We want to prove that $\eta_t(\wt n)\le \eta_t(\wh n)$ holds. Since the total holdings of stocks are the same for both strategies, the self-financing condition (\ref{13.11.2025.3.ck}) implies
that 
 \beam\label{1.9.2026.1}
& & \eta_t(\wh n)-\eta_{t-1}(\wh n) - (\eta_t(\wt n)-\eta_{t-1}(\wt n))\nonumber\\
& & = r_t(\eta_{t-1}(\wh n) - \eta_{t-1}(\wt n)) -(\Pi_t(\wh n)-\Pi_{t-1}(\wh n))+(\Pi_t(\wt n)-\Pi_{t-1}(\wt n)).
\eeam
Applying (\ref{1.9.2026.1}) to all periods up to $t$ and summing up yields
\beam\label{1.9.2026.2}
\eta_t(\wh n)-\eta_t(\wt n)=\sum_{s=1}^tr_s(\eta_{s-1}(\wh n)-\eta_{s-1}(\wt n))-\Pi_t(\wh n)+\Pi_t(\wt n). 
\eeam

{\em Case } $\Pi_t(\wh n)=\Pi_{t-1}(\wh n)$. Then, the assertion follows immediately from (\ref{1.9.2026.1}) and the induction hypothesis for $s=t-1$. 

{\em Case } $\Pi_t(\wh n)>\Pi_{t-1}(\wh n)$. Then, $\Pi_t(\wh n)=\alpha G_t(\wh n)$. On the other hand, we have $\alpha G_t(\wt n)\le \Pi_t(\wt n)$. Together with (\ref{1.9.2026.2}), (\ref{eq:stockgains}) and the definition of $G_t$, it follows that
\beam\label{2.9.2026.1}
\eta_t(\wh n)-\eta_t(\wt n)\ge \sum_{s=1}^t(1-\alpha)r_s(\eta_{s-1}(\wh n)-\eta_{s-1}(\wt n)).
\eeam
Since the RHS of (\ref{2.9.2026.1}) is nonnegative by the induction hypothesis for all $s<t$ and $r\ge 0$, we conclude that
$\eta_t(\wt n)\le \eta_t(\wh n)$ and thus $V^\alpha(x,\wt n)\le V^\alpha(x,\wh n)$. 
\end{proof}

\begin{lemma}\label{Lemma:04.09.2026 ad}
Let $\beta=(\beta_t)_{t=0,\dots,T-1}$ be $\F_T$-measurable random variables taking values in $\bbr_+^d$.
    The mapping $N^{\rm eff}(\beta)(\omega):=N^{\rm eff}(\beta(\omega),\omega)$ defines an $\F_T$-measurable random vector. Furthermore, if $\beta=(\beta_t)_{t=0,\dots,T-1}$ is adapted, then $N^{\rm eff}(\beta)\in \mathcal{N}$.
\end{lemma}
\begin{proof}
Since $(b,\omega)\mapsto N^{\rm eff}(b,\omega)$ is a Carathéodory function, it is $\mathcal{B}(\bbr^{dT})\otimes \F_T$-measurable (see, e.g., Pennanen and Perkki\"o~\cite[Theorem~1.17 and Example 1.12]{pennanen.perkkioe.2024}), and the first claim follows. The second claim follows analogously. 
\end{proof}
For the rest of the paper and by a slight abuse of notation, we write $\beta_t\in\mathcal{R}_t$ for an $\mathcal{F}_t$-measurable 
$\bbr^d_+$-valued random variable~$\beta_t$ if its equivalence class belongs to $\mathcal{R}_t$.
\begin{lemma}\label{lemma:MaxReversible}
There exists an adapted $\bbr^d_+$-valued process~$\wh \beta=(\wh \beta_t)_{t=0,\dots,T-1}$ such that $\wh\beta_t\in\mathcal{R}_t$
for all $t=0,\dots,T-1$ and for all $\beta_t\in \mathcal{R}_t$ we can find $a_t\in L^0(\F_t;\bbr_+)$ with $\beta_t\le a_t \wh \beta_t$ a.s. Consequently, for every adapted process~$\beta$ with $\beta_{t}\in \mathcal{R}_t$ for all $t=0,\dots,T-1$, there exists $a\in L^0(\F_T;\bbr_+)$ such that $\beta_t\le a\wh\beta_t$ for all $t=0,\dots,T-1$ a.s.
\end{lemma}
\begin{proof}
Similar to the proof of Lemma~\ref{lemma:CompactSetA}, we can identify $\mathcal{R}_t$ with an $\F_{t}$-measurable random set~$M_{t}$ such that $\mathcal{R}_t=L^0(\F_t; M_{t})$ (see, e.g., Pennanen and Perkki\"o~\cite[Theorem~1.60]{pennanen.perkkioe.2024}).
Closed-valued random sets have a Castaing representation, meaning that there exists a sequence~$(X_n)_{n\in\bbn}$ of $\F_t$-measurable $\bbr^d$-valued random variables such that $M_t(\omega)=\ov{\{X_n(\omega)\ :\ n\in \bbn\}}$ for all $\omega\in \Omega$ (see, e.g., Pennanen and Perkki\"o~\cite[Theorem~1.27]{pennanen.perkkioe.2024}). Define $X_{\infty}:=0$ and $\tau_j(\omega):=\inf\{n\in \bbn\ :\ X^j_n(\omega)\neq 0\}$ for $j=1,\dots,d$, which are $\F_t-$measurable random variables. 
Then the random variable $\wh\beta_t(\omega):=\sum_{j=1}^dX_{\tau_j(\omega)}(\omega)$ has our desired properties.
\end{proof}

\begin{lemma}\label{2.9.2026.2}
    The terminal wealth $V^\alpha(x,N^{\rm eff}(b,\omega))(\omega)$ is concave in $(x,b)\in \bbr\times\bbr_+^{dT}$ for every $\omega \in \Omega$. Furthermore, for $\omega\in \Omega$ and $b,b^{(1)},b^{(2)}\in \bbr_+^{dT}$ such that $b^{(1)}\le b^{(2)}$ and $b^{(1)},b^{(2)}\in \{(\wt b_t)_{t=0,\dots T-1}\in \bbr_+^{dT}\ :\ \langle \wt b_t,S_{t+1}(\omega)-(1+r_{t+1}(\omega))S_t(\omega)\rangle=0\ \forall t=0,\dots,T-1\}$, we have
\beam\label{eq:monoton in Rev}
V^\alpha(x,N^{\rm eff}(b+b^{(2)},\omega))(\omega)\ge V^\alpha(x,N^{\rm eff}(b+b^{(1)},\omega))(\omega).
\eeam
\end{lemma}
\begin{proof}
    The concavity follows directly from Lemma~\ref{lemma:concave}(iv) and Lemma~\ref{effiz}. Now, let us fix $\omega\in \Omega$ and let $b,b^{(1)},b^{(2)}\in \bbr_+^{dT}$ such that $b^{(1)}\le b^{(2)}$ and $b^{(1)},b^{(2)}\in \{(\wt b_t)_{t=0,\dots T-1}\in \bbr_+^{dT}\ :\ \langle \wt b_t,S_{t+1}(\omega)-(1+r_{t+1}(\omega))S_t(\omega)\rangle=0\ \forall t=0,\dots,T-1\}$. We obviously have $V^\alpha(0,N^{\rm eff}(b^{(2)}-b^{(1)},\omega))(\omega)=0$.
Together with positive homogeneity and concavity of the tax-efficient terminal wealth, we obtain
\beao
 V^\alpha(x,N^{\rm eff}(b+b^{(2)},\omega))(\omega) &=& 2V^\alpha\left(\frac{1}{2}x,N^{\rm eff}(\frac{1}{2}(b+b^{(1)})+\frac{1}{2}(b^{(2)}-b^{(1)}),\omega)\right)(\omega)\\
&\ge&  2\frac{1}{2}V^\alpha(x,N^{\rm eff}(b+b^{(1)},\omega))(\omega)+2\frac{1}{2}V^\alpha(0,N^{\rm eff}(b^{(2)}-b^{(1)},\omega))(\omega) \\
&=& V^\alpha(x,N^{\rm eff}(b+b^{(1)},\omega))(\omega).
\eeao
\end{proof}

\begin{lemma}\label{lemma:OptReversible}
Let $\wh\beta$ be an adapted $\bbr^d_+$-valued process. Then, there exists a Carathéodory function $K:\bbr_+^{dT}\times\Omega\to\bbr_+$ that is $\F_T$-measurable in $\omega$ such that for all $b\in \bbr_+^{dT}$ we have
\beam\label{eq:OptReversible}
V^\alpha(x,N^{\rm eff}(b+K(b,\cdot)\wh\beta(\cdot),\cdot))=\sup_{ k \in \bbr_+}V^\alpha(x,N^{\rm eff}(b+k\wh\beta(\cdot),\cdot))
\eeam
on the set $\{\omega\in\Omega\ :\ \langle\wh\beta_t(\omega),S_{t+1}(\omega)-(1+r_{t+1}(\omega))S_t(\omega)\rangle=0\ \forall t=0\dots,T-1\}=:A$.\\
Furthermore, if $P(A)=1$, then one has $V^\alpha(x,N^{\rm eff}(\beta+K(\beta)\wh\beta))\in \{V^\alpha(x,\wt{N})\ :\ \wt{N}\in\wt{\mathcal{N}}(x)\}$ for every 
adapted $\bbr^d_+$-valued process~$\beta$.
\end{lemma}
\begin{proof}
Let $\wh\beta$ be an $\bbr^d_+$-valued adapted process. First, we observe that 
\beam\label{6.9.2026.1}
V^\alpha(x,N^{\rm eff}(b+k1_A(\omega)\wh\beta(\omega),\omega))(\omega)
& \le & V^0(x,N^{\rm eff}(b+k1_A(\omega)\wh\beta(\omega),\omega))(\omega)\nonumber\\
& = & V^0(x,N^{\rm eff}(b,\omega))(\omega)<\infty\quad \mbox{for all}\ k\in\bbr_+
\eeam
and for all $b\in \bbr^{dT}_+$, $\omega\in\Omega$.
Thus, the supremum in (\ref{eq:OptReversible}) is finite on $A$. By Lemma~\ref{2.9.2026.2}, it can be replaced by a limit, and we can define
\beao
K(b,\omega):=\inf\{k\in\bbr_+ : V^\alpha(x,N^{\rm eff}(b+k\wh\beta(\omega),\omega))(\omega)=\lim_{\wt k\to \infty}V^\alpha(x,N^{\rm eff}(b+\wt k1_A(\omega)\wh\beta(\omega),\omega))(\omega)\}.
\eeao
We fix $\omega \in \Omega$ and set $\wh b:=1_A(\omega)\wh \beta(\omega)$. 
Since equations~(\ref{eq:recursion taxeff1})/(\ref{eq:recursion taxeff2}) involve only finitely many sign conditions, a straightforward induction on $t$ shows that there exist finitely many polyhedral cones $(C_i)_{i=1,\dots,n}$ in $\bbr_+^{dT}$ (not necessarily disjoint) satisfying $\bigcup_{i=1}^n C_i = \bbr^{dT}_+$ and for each $i=1,\dots,n$, there exists a matrix $A^{(i)}\in \bbr^{dT(T+1)/2\times dT}$ such that 
\beam\label{eq:piecewiseLinear1}
N^{\rm eff}(b+k\wh b,\omega)=A^{(i)}(b+k\wh b) \qquad \mbox{for } b+k\wh b\in C_i.
\eeam
We consider $V^\alpha(x,N^{\rm eff}(b+k\wh b,\omega))(\omega)$ as a function in $(b,k)$ denoted by $f(b,k)$. 
As in Proposition~\ref{28.11.2025.1.ck}, we use (\ref{4.10.2025.1.ck}) and observe that for a fixed $\omega\in\Omega$, $V^\alpha(x,\cdot)(\omega)$ is the minimum over $2^T$ affine functions. Combining this with (\ref{eq:piecewiseLinear1}), we get the existence of finitely many $a^j_1,a^{j,i}_3\in \bbr$ and $a_2^j\in \bbr^{dT(T+1)/2}$ such that
\beam\label{eq:piecewise linear}
f(b,k)=\min_{{j\in J}}\left\{ a^j_1+\langle a_2^j,A^{(i)}b\rangle+ka^{j,i}_3\right\}\qquad \mbox{for } b+k\wh b\in C_i,
\eeam
where $J:=\{1,\ldots,2^T\}$ and $a^{j,i}_3=\langle a_2^j,A^{(i)}\wh b\rangle$. By (\ref{eq:piecewiseLinear1}), the RHS of (\ref{eq:piecewise linear})
does not depend on $i$. We show in three steps that $K$ is a Carathéodory function:

{\em Step 1:} Now, we also fix $b\in \bbr_+^{dT}$, and we want to show that $K(b,\omega)<\infty$. Since the $C_i$ are convex cones, the ray $\{b+k\wh b\ :\ k\in \bbr_+\}\subseteq \bbr_+^{dT}$ can leave a cone at most once. Therefore, there exists $i^\star\in \{1,\dots,n\}$ and $k_0\in\bbr_+$ such that $f(b,k)=\min_{j\in J}\{a^j_1+\langle a_2^j,A^{(i^\star)}b\rangle+ka^{j,i^\star}_3\}$ for all $k\ge k_0$. By (\ref{eq:monoton in Rev}), $f(b,k)$ is monotone in $k$, while $\min_{j\in J}\{a^j_1+\langle a_2^j,A^{(i^\star)}b\rangle+ka^{j,i^\star}_3\}$ is the minimum of finitely many affine functions of $k$, hence piecewise affine with only finitely many slope changes. In particular, there exists $\wt k\ge k_0$ and $j^\star$ such that $f(b,k) = a^{j^\star}_1+\langle a_2^{j^\star},A^{(i^\star)}b\rangle+ka^{j^\star,i^\star}_3$ on $[\wt k ,\infty)$. Monotonicity in $k$ implies that $a_3^{j^\star,i^\star}\ge 0$. On the other hand, (\ref{6.9.2026.1}) implies that $a^{j^\star,i^\star}_3\le 0$. Thus, 
$a^{j^\star,i^\star}_3=0$, that is, the function $k\mapsto f(b,k)$ is constant on $[\wt k,\infty)$, and consequently $K(b,\omega)<\infty$.

{\em Step 2:} Let us show that for a fixed $\omega\in \Omega$ the mapping $b\mapsto K(b,\omega)$ is continuous. First observe that by (\ref{eq:piecewise linear}) the family $(f(b,k))_{k\in\bbr_+}$ is uniformly Lipschitz in $b$, and we fix a Lipschitz constant~$L$.

\noindent Case 1: There exists no pair $(j,i)$ such that $a_3^{j,i}>0$. Then by (\ref{eq:piecewise linear})  we get that $K(b,\omega)=0$ for this $\omega\in \Omega$ and for all $b$, thus continuous.\\
Case 2: Suppose there exists at least one pair $(j,i)$ such that $a^{j,i}_3>0$. Since there exist only finitely many pairs $(j,i)$, there exists $(j^{\star\star},i^{\star\star})$ such that $a^{j^{\star\star},i^{\star\star}}_3=\min\{a^{j,i}_3\ :\ (j,i)\text{ s.t. }a^{j,i}_3>0\}>0$. 
Let $b',b\in \bbr_+^{dT}$ be such that $K(b',\omega)\le K(b,\omega)$. By concavity of the function $k\mapsto f(b,k)$ (Lemma~\ref{2.9.2026.2}), it is 
strictly increasing in a neighborhood of each $k<K(b,\omega)$, and by (\ref{eq:piecewise linear}), $a^{j^{\star\star},i^{\star\star}}_3$
is the smallest positive slope. Thus, we obtain
$a^{j^{\star\star},i^{\star\star}}_3(K(b,\omega)-K(b',\omega))\le f(b,K(b,\omega))-f(b,K(b',\omega))$. 
Together with $f(b',K(b,\omega))=f(b',K(b',\omega))$, this yields
\beao
a^{j^{\star\star},i^{\star\star}}_3(K(b,\omega)-K(b',\omega)) & \le & f(b,K(b,\omega))-f(b',K(b,\omega))+f(b',K(b',\omega))-f(b,K(b',\omega))\\
&\le & L|b-b'|+L|b-b'|.
\eeao
By symmetry we arrive at $|K(b,\omega)-K(b',\omega)|\le 2L |b-b'| / a^{j^{\star\star},i^{\star\star}}_3$.

{\em Step 3:} Let us show that for every fixed $b\in \bbr_+^{dT}$ the mapping $\omega\mapsto K(b,\omega)$ is $\F_T$-measurable. Define $h(k,\omega):=\lim_{\wt k\to \infty}V^\alpha(x,N^{\rm eff}(b+\wt k1_{A}(\omega)\wh\beta(\omega),\omega))(\omega)-V^\alpha(x,N^{\rm eff}(b+k\wh\beta(\omega),\omega))(\omega)$ which is an $\F_T$-measurable Carathéodory function by Lemma~\ref{Lemma:04.09.2026 ad}. Hence, by \cite[Example 1.16]{pennanen.perkkioe.2024}, the mapping~$\omega\mapsto S(\omega):=\{k\in \bbr_+\ :\ h(k,\omega)=0\}$ is a closed-valued $\F_T$-measurable random set (cf. Definition~\ref{3.9.2026.1}). Since $K(b,\omega)<\infty$, the set $S(\omega)$ is non-empty for every $\omega\in \Omega$.
Moreover, $K(b,\omega)=\min S(\omega)$ and due to 
$\{\omega\in\Omega\ :\ \min S(\omega)<a\}=\{\omega\in\Omega\ :\ S(\omega)\cap (-\infty,a)\not=\emptyset\}\in \F_T$, $K(b,\cdot)$ is an $\F_T$-measurable random variable.

It remains to show the last assertion of the lemma. Assume that $P(A)=1$ and let $\beta$ be an adapted $\bbr^d_+$-valued process. For each $n\in\bbn$, define $N^n:=N^{\rm eff}(\beta+n\widehat\beta)$. By Lemma~\ref{Lemma:04.09.2026 ad}, we have $N^n\in\mathcal{N}$. Furthermore, define the sets $A_n:=\{\omega\in\Omega:K(\beta(\omega),\omega)\in[n-1,n)\}$ for $n\in\mathbb{N}$. Since $K$ is a Carathéodory function and $\beta$ is $\F_T$-measurable, the mapping
$K(\beta)(\omega):=K(\beta(\omega),\omega)$ is $\F_T$-measurable. Hence, $(A_n)_{n\in\bbn}$ is an $\F_T$-measurable partition of $\Omega$. By (\ref{eq:OptReversible}) and Lemma~\ref{2.9.2026.2}, the strategy $\wt N:=\sum_{n=1}^\infty{1}_{A_n}N^n\in \wt{\mathcal{N}}(x)$ satisfies $V^\alpha(x,\wt N)=V^\alpha(x,N^{\rm eff}(\beta+K(\beta)\wh\beta))$.
\end{proof}
We are now ready to prove Theorem~\ref{main:closedness}.
\begin{proof}[Proof of Theorem \ref{main:closedness}]
    From the definition of $\wt{\mathcal{N}}(x)$, it is clear that we have ${\mathcal{V}^\alpha(x)}\subseteq\{V^\alpha(x,\wt N)\ :\ \wt N\in \wt{\mathcal{N}}(x)\}-L^0(\F_T;\bbr_+)\subseteq \overline{\mathcal{V}^\alpha(x)}$. It is thus enough to prove that under NA the set $\{V^\alpha(x,\wt N)\ :\ \wt N\in \wt{\mathcal{N}}(x)\}-L^0(\F_T;\bbr_+)$ is closed in probability.\\
    Let $(\wt N^n,\wt\mu_n)_{n\in \bbn}\subseteq \wt{\mathcal{N}}(x)\times L^0(\F_T;\bbr_+)$ such that $V^\alpha(x,\wt N^n)-\wt\mu_n\to f$ in probability. By passing to a subsequence, we can and do assume a.s.-convergence. For each $n\in \bbn$, there exists a sequence $(N^{n,m})_{m\in \bbn}\subseteq \mathcal{N}$ with $V^\alpha(x,N^{n,m})\le V^\alpha(x,N^{n,m+1})$ for all $m\in \bbn$ and a partition $(A_{n,m})_{m\in \bbn}\subseteq \F_T$ such that $\wt N^n=\sum_{m=1}^\infty1_{A_{n,m}}N^{n,m}$. Let $m_n:=\inf\{ m\in \bbn\ :\ P(\bigcup_{k=m+1}^\infty A_{n,k})\le 1/n^2\}$. Since $P(V^\alpha(x,N^{n,m_n})<V^\alpha(x,\wt N^{n}))\le 1/n^2$, we know by the a.s.-convergence of $V^\alpha(x,\wt N^n)-\wt\mu_n$ and the Borel-Cantelli lemma that 
    \beam\label{eq:convEndv}
    V^\alpha(x,N^{n,m_n})-\mu_n\to f  \quad\text{a.s.}\quad\mbox{as}\ n\to \infty,
    \eeam
    where $\mu_n:=\wt \mu_n+(V^\alpha(x,N^{n,m_n})-V^\alpha(x,\wt N^{n}))^+$. Next, we decompose the total holdings into the reversible parts~$p^n:=(p_t(\sum_{i=0}^tN^{n,m_n}_{i,t}))_{t=0,\dots,T-1}$ and the purely nonreversible parts~$q^n:=(q_t(\sum_{i=0}^tN^{n,m_n}_{i,t}))_{t=0,\dots,T-1}$. Putting (\ref{eq:convEndv}), Lemma~\ref{lemma:concave}(iii), and Lemma~\ref{lemma:L0-bound nonreversible}(iii) together, we obtain 
that $\sup_{n\in\bbn}|q^n_t|<\infty \mbox{ a.s. for all } t=0,\dots,T-1$. By Lemma~\ref{lemma:komlos}, applied successively to each component of $(q^{n})_{n\in\bbn}$, there exists a sequence of convex weights $(\lambda_n^n,\ldots,\lambda^n_{J_n})_{n\in \mathbb{N}}$ such that
\beam\label{eq:Converstrat2}
 \wt{q}^n_t:=\sum_{k=n}^{J_n}\lambda_k^n q^{k}_t \to \beta_t \quad\text{a.s.}\quad\mbox{as}\ n\to \infty
 \eeam
for some $\beta_t\in L^0(\F_t;\bbr^d_+)$, $t=0,\ldots,T-1$.
The convexity of ${\mathcal{R}}_t$ implies that $\sum_{k=n}^{J_n}\lambda_k^np^k_t\in {\mathcal{R}}_t$. Lemma~\ref{lemma:MaxReversible} gives us an adapted process~$\wh \beta$ and a sequence~$(a^n)_{n\in \bbn}\subseteq L^0(\F_T;\bbr_+)$ such that $\wh \beta_t\in \mathcal{R}_t$  
and $\sum_{k=n}^{J_n}\lambda_k^np^k_t\le a^n\wh\beta_t$ a.s. for all $t=0,\ldots,T-1$. Given $\wh{\beta}$, Lemma~\ref{lemma:OptReversible} provides a Carathéodory function~$K(\cdot,\cdot)$ such that (\ref{eq:OptReversible}) holds. We get
\beam\label{8.9.2026.1}
\sum_{k=n}^{J_n}\lambda_k^n(V^\alpha(x,N^{k,m_k})-{\mu}_k)&\le&\sum_{k=n}^{J_n}\lambda_k^nV^\alpha(x,N^{\rm eff}(q^{k}+p^k))\nonumber\\
&\le &V^\alpha(x,N^{\rm eff}(\sum_{k=n}^{J_n}\lambda_k^n(q^{k}+p^k)))\nonumber\\ 
&=& V^\alpha(x,N^{\rm eff}(\wt q^{n}+\sum_{k=n}^{J_n}\lambda_k^np^k))\nonumber\\
&\le& V^\alpha(x,N^{\rm eff}(\wt q^{n}+a^n\wh\beta))\nonumber\\
&\le& V^\alpha(x,N^{\rm eff}(\wt q^{n}+K(\wt q^n)\wh\beta))\quad\quad \mbox{for all\ }n\in\bbn, 
\eeam
where the first inequality holds by Lemma~\ref{effiz}, the second inequality by the concavity of $V^\alpha(x,N^{\rm eff}(\cdot))$ (Lemma~\ref{2.9.2026.2}), the third inequality by (\ref{eq:monoton in Rev}), and the fourth inequality by (\ref{eq:OptReversible}) applied to $b=\wt{q}^n(\omega)$ for all $\omega\in\Omega$.
By (\ref{eq:convEndv}), the LHS of (\ref{8.9.2026.1}) converges to $f$ a.s.
By (\ref{eq:Converstrat2}) and the continuity of $K(\cdot,\omega)$ and $V^\alpha(x,N^{\rm eff}(\cdot,\omega))(\omega)$, the RHS of (\ref{8.9.2026.1}) converges to $V^\alpha(x,N^{\rm eff}(\beta+K(\beta)\wh\beta))$ a.s., which by Lemma~\ref{lemma:OptReversible} lies in the set $\{V^\alpha(x,\wt N)\ :\ \wt N\in \wt{\mathcal{N}}(x)\}$.
We conclude that $f\in\{V^\alpha(x,\wt N)\ :\ \wt N\in \wt{\mathcal{N}}(x)\}-L^0(\F_T;\bbr_+)$.
\end{proof}

\end{document}